\documentclass[11pt]{article}
\usepackage{enumerate}
\usepackage{url}
\usepackage{amsfonts,amsmath,amssymb,amsthm} 
\usepackage[lmargin=2cm, rmargin=2cm,bottom=2cm,top=2cm]{geometry} 

\usepackage{xspace}
\usepackage[utf8]{inputenc}
\usepackage{tikz}
\usetikzlibrary{arrows,automata}
\usepackage{soul}

\newcommand{\first}{{\rm first}}
\newcommand{\last}{{\rm last}}

\newcommand{\frag}{{\rm Frag}}

\newcommand{\bw}{\ensuremath{\mathbf{w}}\xspace}

\newcommand{\bF}{\ensuremath{\mathbf{F}}\xspace}
\newcommand{\mbf}{\ensuremath{\mathbf{f}}\xspace}
\newcommand{\bB}{\ensuremath{\mathbf{B}}\xspace}

\newcommand{\SbLSP}{S_{\rm bLSP}}
\newcommand{\SRbLSP}{S_{\rm R-bLSP}}

\newcommand{\bT}{\ensuremath{\mathbf{T}}\xspace}

\newtheorem{theorem}{Theorem}[section]

\newtheorem{lemma}[theorem]{Lemma}
\newtheorem{proposition}[theorem]{Proposition}
\newtheorem{corollary}[theorem]{Corollary}

\newtheorem{property}[theorem]{Property}

\theoremstyle{definition}

\newtheorem{example}[theorem]{Example}
\newtheorem{definition}[theorem]{Definition}

\begin{document}
\title{Characterization of infinite LSP words\\and endomorphisms preserving the LSP property}
\author{Gwena\"el Richomme\\
LIRMM, Universit\'e Paul-Val\'ery Montpellier 3, Université de Montpellier, CNRS,\\ Montpellier, France
}

\maketitle

\begin{abstract}
Answering a question of G. Fici, we give an $S$-adic characterization of the
family of infinite LSP words, that is, the family of infinite words having all their left special factors as prefixes.
More precisely we provide a finite set of morphisms $S$ and an automaton ${\cal A}$ such that an infinite word is LSP if and only if it is $S$-adic and one of its directive words is recognizable by ${\cal A}$.
Then we characterize the endomorphisms that preserve the property of being LSP for infinite words.
This allows us to prove that there exists no  set $S'$ of endomorphisms for which the set of infinite LSP words corresponds to the set of $S'$-adic words. 
This implies that an automaton is required no matter which set of morphisms is used.

\end{abstract}

\section{Introduction}

Free monoid morphisms, also sometimes called substitutions, are basic tools to study finite and infinite words.
They are used in various fields
as Combinatorics on Words, Formal Languages or Dynamical Systems (see, {\em e.g.},
\cite{Berthe_Rigo2010CANT,Lothaire1983book,Lothaire2002,Pytheas2002,RozenbergSalomaa1997Handbook}).
Literature contains many examples of interesting infinite words that
are fixed points of endomorphisms. Their success and interest are mainly
due to the simplicity of their definitions and the ability given by endomorphisms to prove their properties.
Let $\bw$ be such a fixed point of morphism
and let $f$ be a morphism 
such that $\bw = f(\bw)$.
This word $\bw$ can be seen as the limit $\lim_{n \to \infty} f^n(a)$ where $a$ is the first letter of $\bw$ (with some extra needed conditions on $f$ to ensure that this limit exists and is an infinite word).
This word $\bw$ can also be considered as a word which can be recursively
desubstituted by the morphism $f$: here \textit{desubstituted} means the existence of a word $\bw_1$ such that $\bw = f(\bw_1)$ and \textit{recursively desubstituted} means that $\bw_1$ can be itself desubstituted by $f$ and so on).

Paraphrasing \cite{Berthe_Delecroix2014RIMS}, if one wants to go beyond the morphic case, and thus, 
get more flexibility in the hierarchical structure, 
one might want to change the morphism at each desubstitution step (the point of view in \cite{Berthe_Delecroix2014RIMS} is somewhat dual).
Thus considering a set $S$ of morphisms instead of a single morphism $f$, 
we may consider the $S$-adic words, that is, words that can be recursively desubstituted over $S$. A more precise definition is given in Section~\ref{sec:s-adicity}.
For more information on $S$-adic words, readers can consult, \textit{e.g.}, papers \cite{Berthe2016RIMS,Berthe_Delecroix2014RIMS}, Chapter 12 in \cite{Pytheas2002} and their references. Note that we will sometimes write $S$-adic for substitutive-adic word without direct reference to a set $S$ of morphisms.

$S$-adicity arises naturally in various studies as, for instance, 
those of Sturmian words \cite[Chap.2]{Lothaire2002}\cite[Chap. 5]{Pytheas2002}
or of Arnoux-Rauzy words \cite{ArnouxRauzy1991}.
In \cite{Ferenczi1996ETDS} (see also \cite{Leroy2012thesis,Leroy2014DMTCS})
Ferenczi  proved that
any word with factor complexity bounded by some affine function
is $S$-adic for some finite set $S$ of morphisms that
depends on the bound of the first difference of the factor complexity 
(remember that, for any word with factor complexity bounded by some affine function,
the first difference of the factor complexity is bounded by a constant \cite{Cassaigne1996DLT}). 
For each of the previous examples,
the $S$-adic properties are not characteristic of the considered families.
Episturmian words that generalize Sturmian and Arnoux-Rauzy words are the $S_E$-adic words where $S_E$ is the finitely generated set of morphisms that preserve episturmian words \cite{DroubayJustinPirillo2001,JustinPirillo2002,Richomme2003TCS}.
Over the binary alphabet, it is an exercise to show that the set of balanced words (Sturmian words and ultimately periodic balanced words) are the $S_S$-adic words where
$S_S$ is the finitely generated set of morphisms that preserve Sturmian words.
As far as the author knows it, the kind of characterizations of the two previous examples is very rare.

$S$-adic characterizations of families of words is generally twofold.
In addition to the induced set $S$ of morphisms, a characterization of
allowed infinite sequences of desubstitutions is provided.
For instance there exist two disjoint sets $S_a$ and $S_b$ of morphisms
such that a word is Sturmian if and only it can be recursively decomposed over $S_a\cup S_b$ using infinitely often elements from $S_a$ and 
infinitely often elements from $S_b$ \cite{BertheHoltonZamboni2006}.
This condition can be described using infinite paths with prohibited segments in an automaton or a graph. 
These kinds of conditions are also used, for instance, in the characterization of words for which the first difference of factor complexity is bounded by 2  \cite{Leroy2012thesis,Leroy2014DMTCS}
or in the characterization of sequences arising from the study of the Arnoux-Rauzy-Poincar\'e multidimensional continued fraction algorithm \cite{Berthe_Labbe2015AAM}. For the last example, all infinite paths in the graphs are allowed.

Extending an initial work by M.~Sciortino and L.Q.~Zamboni \cite{Sciortino_Zamboni2007DLT},
G.~Fici investigated relations between the structure of the suffix automaton built from a finite word $w$ and the combinatorics of this word \cite{Fici2011TCS}. 
He proved that the words having their associated automaton with a minimal number of states (with respect to the length of $w$) are the words having all their left special factors as prefixes. 
G.~Fici asked in the conclusion of his paper for a characterization of the set of words having the previous property, that he called the LSP property, both in the finite and the infinite case. 
In this paper, we provide an $\SbLSP$-adic characterization
of LSP infinite words (for a suitable $\SbLSP$ set of morphisms) 
using an automaton recognizing allowed infinite desubstitutions over
$\SbLSP$.
We prove that there exists no set of morphisms $S$
such that the family of LSP words 
is the family of $S$-adic words.

Our $S$-adic characterization is a refinement
of the one presented at Conference DLT 2017 \cite{Richomme2017DLT}.
Main ideas of Sections~\ref{sec:morphisms} to \ref{sec:carac} were already presented and used in \cite{Richomme2017DLT}.
But the set of morphisms considered here allows to provide a smaller automaton 
(in particular it can be drawn for a	 three-letter alphabet while this was not possible in \cite{Richomme2017DLT}) even
if this set is larger as it is 
the set of morphisms considered in \cite{Richomme2017DLT} plus their restrictions to smaller alphabets.
The proof that one cannot have a characterization without restrictions on allowed desubstitutions (Section~\ref{sec:preservation}) is new.

The paper is organized as follows.
After introducing in Section~\ref{sec:morphisms} our basis of morphisms $\SRbLSP$,
in Section~\ref{sec:s-adicity}, we show that all infinite LSP words are $\SRbLSP$-adic.
Section~\ref{sec:fragility} introduces a property of infinite LSP words and a property of morphisms in $\SbLSP$ that together allow to explain why the LSP property is lost when applying an LSP morphism to an infinite LSP word.
Section~\ref{sec:origin} allows to trace the origin of the previous property of infinite LSP words.
Based on this information, Section~\ref{sec:automaton} defines our automaton and
Section~\ref{sec:carac} proves our characterization of infinite LSP words.
In Section~\ref{sec:preservation} we characterize endomorphisms preserving LSP words
and deduce that the set of LSP words cannot be characterized as a set of $S$-adic words, whatever $S$ is.

\section{\label{sec:morphisms}Some basic morphisms}

We assume that readers are familiar with combinatorics on words; for omitted definitions (as for instance, factor, prefix, ...) see, \textit{e.g.}, \cite{Berthe_Rigo2010CANT,Lothaire1983book,Lothaire2002}. 
Given an alphabet $A$, $A^*$ is the set of all finite words over $A$, including the empty word $\varepsilon$, and $A^\omega$ is the set of all infinite words over $A$. For a non-empty word $u$, let $\first(u)$ denote its first letter, $\last(u)$ its last letter and ${\rm alph}(u)$ its set of letters. Notations $\first$ and $\rm alph$ are similarly defined for infinite words.

A finite word $u$ is a \textit{left special factor} of a finite or infinite word $w$ 
if there exist at least two distinct letters $a$ and $b$ such that both words $au$ and $bu$ occur in $w$. Following G.~Fici \cite{Fici2011TCS}, a finite or infinite word $\bw$ is \textit{LSP} if all its left special factors are prefixes of $\bw$. 
We also say say $\bw$ has \textit{the LSP property}.
We study these words using morphisms.

Given two alphabets $A$ and $B$, a \textit{morphism} (\textit{endomorphism} when $A = B$) $f$ is a map from $A^*$ to $B^*$ such that for all words $u$ and $v$ over $A$, $f(uv) = f(u)f(v)$.
Morphisms are entirely defined by images of letters.
Morphisms extend naturally to infinite words.
We consider only nonerasing morphisms, that is, morphisms $f$ such that $f(x) = \varepsilon$ implies $x = \varepsilon$.
When $X$ and $Y$ are two sets of morphisms, let $XY$ denote the set of morphisms $f \circ g$ with $f \in X$, $g \in Y$ such that if $f$ is from $A^*$ to $B^*$ and $g$ is from $C^*$ to $D^*$ then $D \subseteq A$.

We call \textit{basic LSP morphism} on an alphabet $A$, or $\textit{bLSP}$ in short, 
any endomorphism $f$ of $A^*$ verifying:
\begin{itemize}
\itemsep0cm
\item there exists a letter $\alpha$ such that $f(\alpha) = \alpha$, and
\item for all letters $\beta \neq \alpha$, there exists a letter $\gamma$ such that $f(\beta) = f(\gamma) \beta$
\end{itemize}

From this definition one can naturally associate the rooted tree 
(all labels are on vertices and distinct vertices have distinct labels)
whose vertices are elements of $A$, whose root is $\alpha$ 
and whose (oriented) edges are pairs $(\beta, \gamma)$ of letters such that
$f(\beta) = f(\gamma)\beta$.
Conversely,
given a labeled rooted tree $T = (A, E)$, 
let $f_T$ be the morphism defined by: for all letters $\beta$, $f_T(\beta)$ is the word obtained concatenating vertices on the path in $T$ from the root of $T$ to $\beta$. This morphism is bLSP.

Let $\SbLSP(A)$ (or shortly $\SbLSP$ when $A$ is clear) denote the set of all bLSP morphisms over the alphabet $A$. 
We have just seen that there is a bijection between this set and the set of labeled rooted trees with label in $A$.
Thus, denoting by $\#X$ the cardinality of a set $X$, there are $(\#A)^{\#A-1}$ elements in $\SbLSP(A)$ 
(see Sequence A000169 in The On-Line Encyclopedia of Integer Sequences: this sequence enumerates rooted trees; its first values are $1$, $2$, $9$, $64$, $625$, $7776$, $117649$, $2097152$). 

When $f$ is a morphism defined on an ordered alphabet $A = \{a_1, \ldots, a_k\}$, 
we let $[ u_1, u_2, \ldots, u_k]$ denote the morphism defined by $a_1 \mapsto u_1$, $a_2 \mapsto u_2$, \ldots, $a_k \mapsto u_k$. 
For instance $[a, ab, ac, acd, ace]$  defines the bLSP morphism $f$ such that $f(a) = a$, $f(b) = ab$, $f(c) = ac$, $f(d) = acd$, $f(e) = ace$. The rooted tree associated with $f$ is given in Figure~\ref{Ex1}.

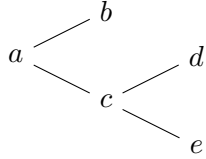
\begin{figure}[!ht]
\begin{center}
\begin{tikzpicture}[scale=0.3]
\begin{scope}
\node (a) at (0,0) {$a$};
\node (b) at (4,2) {$b$};
\node (c) at (4,-2) {$c$};
\node (d) at (8,0) {$d$};
\node (e) at (8,-4) {$e$};
\draw (a) -- (b);
\draw (a) -- (c);
\draw (c) -- (d);
\draw (c) -- (e);
\end{scope}
	
\end{tikzpicture}
\end{center}
\caption{\label{Ex1}Rooted tree associated with the morphism $[a, ab, ac, acd, ace]$}
\end{figure}

\begin{example}
\label{Exemple1}
\rm 
Some examples of bLSP morphisms are the standard episturmian morphisms $L_\alpha$ defined, for $\alpha$ a letter, by: $L_\alpha(\alpha) = \alpha$ and $L_\alpha(\beta) = \alpha \beta$ for any letter $\beta \neq \alpha$ (see for instance \cite{JustinPirillo2002,Richomme2003TCS} for example of uses of these morphisms). The associated rooted tree has its root connected to all other vertices. Observe that $\SbLSP(\{a, b\}) = \{ L_a, L_b\} = \{ [a, ab], [ba, b]\}$.
These morphisms are well-known in the context of Sturmian words. 
They are denoted $\tau_a$ and $\tau_b$ in \cite{BertheHoltonZamboni2006} from which it can be seen that standard Sturmian words are the non-periodic $\{\tau_a, \tau_b\}$-adic words (see also \cite{LeveRichomme2007TCS}). 
\end{example}

\begin{example}
\label{Exemple2}
\rm
For any alphabet $A_k = \{a_1, \ldots, a_k\}$, let $\lambda_{a_1\cdots a_k}$ be the morphism defined by $\lambda_{a_1\cdots a_k}(a_i) = a_1\cdots a_i$. The associated rooted tree is a path. For instance $\lambda_{acb} = [a, acb, ac]$. 
Note that $\SbLSP(\{a, b, c\}) = \{ L_a,$ $L_b,$ $L_c,$ $\lambda_{abc},$ 
$\lambda_{acb},$ $\lambda_{bac},$ $\lambda_{bca},$ $\lambda_{cab},$ $\lambda_{cba}\}$. Observe that these morphisms are the mirror morphisms of Arnoux-Rauzy and Poincar\'e morphisms (here $f$ is a \textit{mirror morphism} of $g$ if $f(a)$ is the mirror image or reversal of $g(a)$ for all letters $a$) used by V.~Berth\'e and S.~Labb\'e \cite{Berthe_Labbe2015AAM}.
\end{example}

Let $A$, $B$ and $A'$ be three alphabets with $A' \subseteq A$ and let $f$ be a morphism from $A^*$ to $B^*$.
The restriction of $f$ to $A'$ is the morphism from $A'^*$ to ${\rm alph}(f(A'))^*$ denoted $f_{|A'}$ or $f{|A'}$and defined by $f_{|A'}(\alpha) = f(\alpha)$ for each letter $\alpha$ of $A'$.
An R-bLSP morphism is any restriction to a subalphabet of a morphism in $\SbLSP$.
Let $\SRbLSP(A)$ or simply $\SRbLSP$ denote the set of R-bLSP morphisms. This means that if $f$ belongs to $\SRbLSP(A)$, there exists a subalphabet $B$ of $A$ and an element $g$ of $\SbLSP$ such that $f = g_{|B}$.

By construction of an $R$-bLSP morphism $f$, all images of letters by $f$ begin with the same letter: let $\first(f)$ denote it.
Next properties are also direct consequences of the definition of R-bLSP morphisms. They will often be used without explicit mention.

\begin{property}
\label{prop:bLSP properties}
Let $f$ be an R-bLSP morphism over the alphabet $A$.
\begin{enumerate}
%\item there exists a unique letter $\alpha \in A$ such that for all $\beta \in A$, $\first(f(\beta)) = \alpha$;
\item \label{prop:last} for all $\beta \in A$, $\last(f(\beta)) = \beta$;
%\item \label{rem:unique letter} there exists a unique letter $\alpha \in A$ such that $f(\alpha) = \alpha$: $\alpha = \first(f)$;
\item $f(A)$ is a suffix code (no word of $f(A)$ is a suffix of another word in $f(A)$);
\item $f$ is injective both on the set of finite words and the set of infinite words;
\item for all $\beta \in A$, $x$, $y \in A^*$, if $|x| = |y|$ and if $x\beta$ and $y\beta$ are  factors of words in $f(A)$, then $x = y$;
\item \label{property1} for all letters $\beta$, $\gamma$, $|f(\beta)|_\gamma \leq 1$.
\end{enumerate}

\end{property}

\section{\label{sec:s-adicity}Substitutive-adicity of infinite LSP words}

Let $S$ be a set of morphisms. 
Usually an infinite word $\bw$ is said to be \textit{$S$-adic}
if there exist a sequence $(f_n)_{n\geq 1}$ of morphisms in $S^\omega$
and a sequence of letters $(a_n)_{n\geq 1}$ such that 
$\lim_{n \to +\infty} |f_1f_2\cdots f_n(a_{n+1})| = +\infty$ and
$\bw = \lim_{n \to +\infty} f_1f_2\cdots f_n(a_{n+1})$.
The sequence $(f_n)_{n\geq 1}$ is called the \textit{directive word} of $\bw$. 
We consider here $S$-adicity in a rather larger way:
a word $\bw$ is \textit{$S$-adic} with directive word $(f_n)_{n\geq 1}$ if there exists an infinite sequence of infinite words $(\bw_n)_{n \geq 1}$ such that $\bw_1 = \bw$ and $\bw_n = f_n( \bw_{n+1})$ for all $n \geq 1$. If the former definition is verified, the latter is also verified. 
This second definition may include degenerated cases as, for instance, the word $a^\omega$ that is $\{Id\}$-adic with $Id$ the identity morphism. 
To understand $S$-adicity, 
it may be better to see $\bw_{n+1}$ as the inverse image of $\bw_n$,
rather than seing $\bw_n$ as the image of $\bw_{n+1}$.
As morphisms are sometimes called substitutions, we 
say that $\bw_{n+1}$ is a desubstituted word from $\bw_n$
and so saying that $\bw$ is $S$-adic means that $\bw$ can be recursively desubstituted (using elements of $S$).
Sometimes $f_n$ can be defined on a larger alphabet than ${\rm alph}(\bw_{n+1})$.
In this case $f_n$ can be replaced with its restriction to ${\rm alph}(\bw_{n+1})$.
We will say that $(f_n)_{n \geq 1}$ is a \textit{fitted directive} word when for all $n \geq 1$, $f_n$ is defined on ${\rm alph}(\bw_{n+1})$.

\begin{proposition}
\label{characLSPwords}Any infinite LSP word is $\SRbLSP$-adic and so is $\SbLSP$-adic. 
Moreover any LSP word has exactly one fitted directive word over $\SRbLSP$.
\end{proposition}

The proof of this proposition follows the following general scheme of proof and so is a direct consequence of the next two lemmas.

Given a set $S$ of morphisms, in order to prove that infinite words verifying a property $P$ are $S$-adic, it suffices to prove that for all infinite words $\bw$ verifying $P$,
\begin{enumerate}
\item there exist $f \in S$ and an infinite word $\bw'$ such that $\bw = f(\bw')$, and
\item if $\bw = f(\bw')$ with $f \in S$, then $\bw'$ verifies Property $P$.
\end{enumerate}

To prove the uniqueness of directive words one has also to prove the uniqueness of $f$ and $\bw'$ in the first item above.

\begin{lemma}
\label{lemma2}
Given any finite or infinite LSP word $\bw$, 
there exist a unique infinite word $\bw'$ and a unique morphism $f$ in $\SRbLSP(alph(\bw'))$ such that $\bw = f(\bw')$.
\end{lemma}
 
\begin{proof}
Let $\bw$ be a non-empty finite or infinite LSP word and let $\alpha$ be its first letter.
Let $X$ be the set of words over ${\rm alph}(\bw)\setminus\{\alpha\}$ such that $\bw$ can be factorized over $\{\alpha\} \cup X$.
Let $G$ be the graph $({\rm alph}(\bw), E)$ with $E$ the set of edges $(\beta, \gamma)$ such that $\beta\gamma$ is a factor of a word $\alpha u$ with $u \in X$. By LSP Property of $\bw$, each letter occurring in a word of $X$ is not left special in $\bw$.
Hence $G$ is a rooted tree with $\alpha$ as root, that is, for any letter $\beta$ in ${\rm alph}(\bw)$, there exists a unique path from $\alpha$ to $\beta$. 
Let $u_\beta$ denote the word obtained by concatenating the letters occurring in the path. 
Let $g$ be the morphism defined by $g(\beta) = u_\beta$ for all $\beta$ in ${\rm alph}(\bw)\setminus\{\alpha\}$, and by $g(\alpha) = \alpha$.
By construction, $g$ is bLSP.

As $\alpha$ does not occur in words of $X$ 
and as $\bw$ begins with $\alpha$,
$\bw$ has a unique decomposition over $\alpha X \cup \{\alpha\}$.
If this decomposition holds over $\alpha X$, set $Y = \alpha X$.
Otherwise set $Y = \alpha X \cup \{\alpha\}$.
Observe that $Y \subseteq \{ u_\beta \mid \beta \in {\rm alph}(\bw)\}$.
Let $f$ be the restriction of $g$ to the set $\{ \beta \mid u_\beta \in Y\}$: $f \in S_{\rm R-bLSP}$.
By construction $\bw = f(\bw')$ for a word $\bw'$ and $f \in S_{\rm R-bLSP}({\rm alph}(\bw'))$.

Assume $\bw = h(\bw'')$ with $h$ an R-bLSP morphism and $\bw''$ an infinite word. 
Let $A' = {\rm alph}(\bw')$ and $A'' = {\rm alph}(\bw'')$.
By the definition of R-bLSP morphisms, all images of letters by $f$ or $h$ begin with the same letter $\alpha$ 
(which is also the first letter of $\bw$).
Still by the definition of LSP morphisms, this letter $\alpha$ occurs only once in each image of letters by $f$ and $h$.
Positions of the letter $\alpha$ in $\bw$ determine the beginnings of images of letters by $f$ and $h$.
This implies that $f(A') = h(A'')$.
By Property~\ref{prop:bLSP properties}(\ref{prop:last}),
for letters $\beta$ and $\gamma$, $\last(f(\beta)) = \gamma$ or $\last(h(\beta)) = \gamma$ implies $\beta = \gamma$. 
Hence $f = h$ and $\bw' = \bw''$.
\end{proof}

\begin{lemma}
\label{lemma1}
For any R-bLSP morphism $f$ and any infinite word $\bw$, if $f(\bw)$ is LSP then $\bw$ is LSP.
\end{lemma}

\begin{proof}
Assume by contradiction that $\bw$ is not LSP. 
This means that $\bw$ has (at least) one left special factor that is not one of its prefixes.
Considering such a factor of minimal length,
there exist a word $u$ and letters $a$, $b$, $\beta$, $\gamma$ such that
$a \neq b$, $\beta \neq \gamma$,
$ua$ is a prefix of $\bw$, 
$\beta ub$ and $\gamma ub$ are factors of $\bw$.
As $f$ is an R-bLSP morphism, 
there exists a unique letter $\alpha$ which is the first letter of all non-empty images by $f$.
The word $f(u)f(a)\alpha$ is a prefix of $f(\bw)$.
Moreover by Property~\ref{prop:bLSP properties}(\ref{prop:last}), 
the words $\beta f(u)f(b)\alpha$ and $\gamma f(u)f(b)\alpha$
are factors of $\bw$ (here the fact that $\bw$ is infinite is useful: each factor is followed by a letter whose image begins with $\alpha$).
As $f(a) \neq f(b)$ and as the letter $\alpha$
occurs only as a prefix in $f(a)$ and $f(b)$, 
$f(a)\alpha$ is not a prefix of $f(b)\alpha$ and, conversely, 
$f(b)\alpha$ is not a prefix of $f(a)\alpha$.
Hence there exist a word $v$ and letters $a'$, $b'$ such that $a' \neq b'$, 
$va'$ and $vb'$ are respectively prefixes of $f(a)\alpha$ and $f(b)\alpha$.
It follows that $f(u)va'$ is a prefix of $f(\bw)$ while
$\beta f(u)vb'$ and $\gamma f(u)vb'$
are factors of $f(\bw)$: $f(\bw)$ is not LSP, a contradiction.
\end{proof}

Observe that Lemma~\ref{lemma1} does not hold for finite words. For instance the word $baa$ is not LSP while its image $abaa$ by the morphism $[a, ab]$ is LSP. 

\section{\label{sec:fragility}Fragilities of infinite LSP words}

The converse of Lemma~\ref{lemma1} is false: there exist an infinite LSP word $\bw$ and a bLSP morphism $f$ such that $f(\bw)$ is not LSP.

\begin{example}
\label{Exemple3}
\rm
Let $\bF$ be the well-known Fibonacci word (the fixed point of the endomorphism $[ab, a]$), 
and let $g = \lambda_{acb} = [a, acb, ac]$. 
The word $g^2(\bF)$ begins with the word $g^2(ab) = g(aacb) = aaacacb$ that contains the factor $ac$ which is left special but not a prefix of the word.
Hence the word $g^2(\bF)$ is not LSP while $\bF$ is LSP and $g$ is bLSP (actually one can prove, using Lemma~\ref{when (a, b, c)-breaking morphisms appear} below, that $g(\bF)$ is LSP). 
\end{example}

In what follows, we introduce some properties of LSP words and morphisms that explain in which context a (breaking) R-bLSP morphism can map a (fragile) infinite LSP word on a non-LSP word.

\begin{definition}
\label{def_fragility}\rm
Let $a, b, c$ be three pairwise distinct letters. 
An infinite word $\bw$ is $(a, b, c)$-\textit{fragile} 
if there exist a word $u$ and distinct letters $\beta$ and $\gamma$
such that the word $ua$ is a prefix of $\bw$ and the words $\alpha u b$ 
and $\beta u c$ are factors of $\bw$. The word $u$ is called an $(a, b, c)$-\textit{fragility} of $\bw$.
We will also say that, when we need letters $\alpha$ and $\beta$, $\bw$ is $(a, b, c, \beta, \gamma)$-fragile and the word $u$ is an $(a, b, c, \beta, \gamma)$-\textit{fragility} of $\bw$.
\end{definition}

For instance, the empty word $\varepsilon$ is an $(a, b, c, c, a)$-fragility of $g(\bF)$: $\varepsilon a$ is a prefix of $g(\bF) = aacb\cdots$ while $c\varepsilon b$ and $a\varepsilon c$ are factors of  $g(\bF)$.
More generally any factor $abc$ or $acb$ in an infinite word
produces an $(a,b,c)$-fragility. 
One can also observe that, by symmetry of the definition, any $(a, b, c)$-fragile word is also $(a, c, b)$-fragile. Finally let us note that no fragility exists in words over two letters (as the definition needs three pairwise different letters).

The main idea of introducing the previous notion is that for any $(a, b, c)$-fragile LSP word $\bw$, there exists a bLSP morphism such that $f(\bw)$ is not LSP.
For instance, if $u, \alpha, \beta, \bw$ are as in Definition~\ref{def_fragility} and if $g = \lambda_{acb} = [a, acb, ac]$, the word $g(u)aa$ is a prefix of $g(\bw)$ whereas the words $\alpha g(u)acb$ and $\beta g(u)ac$ are factors of $g(\bw)$, so that $g(\bw)$ is not LSP since $g(u)ac$ is left special but not a prefix of $g(\bw)$.

\begin{definition}\rm
Let $a, b, c$ be three pairwise distinct letters.
A morphism $f$ is LSP $(a, b, c)$-\textit{breaking}, if for all $(a, b, c)$-fragile LSP word $w$, $f(w)$ is not LSP.
\end{definition}

For instance, the morphism $\lambda_{acb} = [a, acb, ac]$ is $(a, b, c)$-breaking.

\begin{lemma}
\label{when (a, b, c)-breaking morphisms appear}
Let $\bw$ be an infinite LSP word and let $f$ be an R-bLSP morphism defined on ${\rm alph}(\bw)$.
The following assertions are equivalent:
\begin{enumerate}
\item The word $f(\bw)$ is not LSP;
\item There exist some pairwise distinct letters $a, b, c$ such that $\bw$ is $(a, b, c)$-fragile and 
the longest common prefix of $f(b)$ and $f(c)$ is  strictly longer than the longest common prefix of $f(a)$ and $f(b)$;
\item There exist some pairwise distinct letters $a, b, c$ in $\bw$ such that $\bw$ is $(a, b, c)$-fragile 
and $f$ is LSP $(a, b, c)$-breaking.
\end{enumerate}
\end{lemma}

\begin{proof}
$1 \Rightarrow 2$.
Assume first that $f(\bw)$ is not LSP. 
There exists a left special factor $V$ of $f(\bw)$ which is not a prefix of $f(\bw)$.
Let $v$ be the longest common prefix of $V$ and $f(\bw)$.
Let $a', b'$ be the letters
such that $va'$ is a prefix of $f(\bw)$ 
and $vb'$ is a prefix of $V$: by construction $a' \neq b'$.
Let also $\beta$, $\gamma$ be distinct letters such that $\beta V$ and $\gamma V$ are factors of $f(\bw)$ (also $\beta v b'$ and $\gamma v b'$ are factors of $f(\bw)$).

By the definition of R-bLSP morphisms, 
the letter $\alpha = \first(f)$ is
the unique letter that can be left special in $f(\bw)$.
As $\alpha$ is the first letter of $f(\bw)$, we have $v \neq \varepsilon$ and $\first(v) = \first(f)$.
As $\alpha$ occurs exactly at the first position in all images of letters, 
occurrences of $\alpha$ mark the beginning of images of letters in $f(\bw)$.
Considering the last occurrence of $\alpha$ in $v$,
we can write $v = f(u)\alpha x$ with $|x|_\alpha = 0$.
Let $a$, $b$, $c$ be letters such that:
\begin{itemize}
\item $ua$ is a prefix of $\bw$, and, $va' = f(u)\alpha x a'$ is a prefix of $f(ua)$ when $a' \neq \alpha$ or $v = f(ua)$ when $a' = \alpha$;
\item $\beta ub$ is a factor of $\bw$, and, $\beta vb'$ is a prefix of $\beta f(ub)$ when $b' \neq \alpha$ or $v = f(ub)$ when $b' = \alpha$;
\item $\gamma uc$ is a factor of $\bw$, and, $\gamma vb'$ is a prefix of $\gamma f(uc)$  when $b' \neq \alpha$ or $v = f(uc)$ when $b' = \alpha$.
\end{itemize}

As $a' \neq b'$, we have $a \neq b$ and $a \neq c$.
Observe that until now we did not use the fact that $\bw$ is LSP.
This implies $b \neq c$ (and so $b' \neq \alpha$).
Indeed otherwise $ub$ would be a left special factor of $\bw$ without being one of its prefixes: a contradiction with the fact that $\bw$ is an LSP word.
Thus $\bw$ is $(a, b, c)$-fragile.

This ends the proof of Part $1 \Rightarrow 2$ as $\alpha xb'$ is a common prefix of $f(b)$ and $f(c)$ and $\alpha x$ is the longest common prefix of $f(a)$ and $f(b)$.

$2 \Rightarrow 3$.
By hypothesis, $f(a) = v \delta w_1$, $f(b) = v \gamma w_2$ and $f(c) = v \gamma w_3$ for letters $\delta, \gamma$ 
and words $w_1$, $w_2$ and $w_3$ with $\delta \neq \gamma$.
Let $\bw'$ be any LSP $(a, b, c)$-fragile infinite word. 
Let $u'$, $\beta'$ and $\gamma'$ be the word and letters such that 
$u'a$ is a prefix of $\bw'$ while $\beta' u' b$ and $\gamma' u' c$ are factors of $\bw'$ with $\beta' \neq \gamma'$.
The word $f(\bw')$ has $f(u')v\delta$ as a prefix and words $\beta'f(u')v\gamma$ and $\gamma'f(u')v \gamma$ as factors.
As $\delta \neq \gamma$, the word $f(\bw')$ is not LSP. The morphism $f$ is LSP $(a,b,c)$-breaking.

$3 \Rightarrow 1$.
This follows the definition of $(a, b, c)$-fragile words and LSP $(a, b, c)$-breaking morphisms.
\end{proof}

Observe that we have also proved the next result.

\begin{corollary}
\label{carac LSP breaking morphisms}
An R-bLSP morphism is LSP $(a, b, c)$-breaking for pairwise distinct letters $a$, $b$ and $c$ if and only if the longest common prefix of $f(b)$ and $f(c)$ is strictly longer than the longest common prefix of $f(a)$ and $f(b)$.
\end{corollary}

To end this section let us mention that in the binary case the converses of Lemma~\ref{lemma1} and Proposition~\ref{characLSPwords} hold.
Indeed as shown in Lemma~\ref{when (a, b, c)-breaking morphisms appear}, if $\bw$ is LSP, $f$ is an R-bLSP morphism and $f(\bw)$ is not LSP, then $\bw$ contains at least three distinct letters. 
Recall that the elements of $\SRbLSP(\{a, b\})$ are the morphisms $L_a$, $L_b$ and their restrictions to alphabets $\{a\}$ and $\{b\}$.

\begin{corollary}
\label{C:carac1_binary}
If $\bw$ is an LSP infinite word over $\{a, b\}$ and if $f$ is defined over ${\rm alph}(\bw)$ and
belongs to the set $\SRbLSP(\{a, b\}$ then $f(\bw)$ is also LSP. 
Consequently an infinite word over $\{a, b\}$ is LSP 
if and only if it is $\{L_a, L_b \}$-adic.
\end{corollary}

\section{\label{sec:origin}Origin of fragilities}

Before characterizing infinite LSP words, we need to know how fragilities in an LSP word $\bw$ can appear. 
We will see that this depends only on the morphisms in the fitted $\SRbLSP$-adic directive word of $\bw$ (Theorem~\ref{th_carac}). At a first step, let $f$ be an R-bLSP morphism from $A^*$ to $B^*$ with $\alpha = \first(f)$. We examine the fragilities occurring in the image by $f$ of an LSP word.

\textit{New fragilities.} 
Assume that for some letters $b$, $c$, $\beta$, $\gamma$ 
with $\alpha \neq b \neq c \neq \alpha$ and $\beta \neq \gamma$, the words $\beta b$ and $\gamma c$ belong to $f(A)$, that is, are factors of images of letters by $f$. Then for any word $\bw$ containing all letters of $A$ (or at least the letters whose images contain $\beta b$ and $\gamma c$), the empty word is an $(\alpha, b, c, \beta, \gamma)$-fragility of $f(\bw)$. We say that $\varepsilon$ is an $(\alpha, b, c,  \beta, \gamma)$-fragility (or simply an $(\alpha, b, c)$-fragility) \textit{associated with $f$}. For instance, $(a, b, c)$ is a fragility associated with the morphism $\lambda_{abc} = [a, ab, abc]$.

\textit{Propagated fragilities.} 
Assume now that an infinite word $\bw$ over $A$ contains an $(a', b', c', \beta, \gamma)$-fragility $v$.
Observe that: 
$f(va')\alpha$ is a prefix of $f(\bw)$;
$\beta f(vb')\alpha$ and $\gamma f(vc')\alpha$ are factors of $f(\bw)$; 
the word $f(v)\alpha$ is a prefix of the three words $f(va')\alpha$, $f(vb')\alpha$
and $f(vc')\alpha$. 
Assume there exist a common prefix $u$ of these words and pairwise distinct letters $a$, $b$ and $c$ 
such that $ua$ is a prefix of $f(va')\alpha$, 
$\beta ub$ is a prefix of $\beta f(vb')\alpha$
and $\gamma uc$ us a prefix of $\gamma f(vc')\alpha$. 
Then $u$ is an $(a, b, c, \beta, \gamma)$-fragility in $f(w)$. 
We say that this fragility is \textit{propagated} by $f$ from the $(a', b', c', \beta, \gamma)$-fragility $v$. 

Note that $|u| \geq |f(v)\alpha|$ and more precisely $u = f(v)\alpha u'$ for some word $u'$.
Observe that 
$\alpha u'a$ is a prefix of $f(a')\alpha$,
$\alpha u'b$ is a prefix of $f(b')\alpha$ and
$\alpha u'c$ is a prefix of $f(c')\alpha$.
This is an important fact as it shows that letters $a$, $b$ and $c$ depend only on $f$ and on letters $a'$, $b'$ and $c'$:
if a word contains an $(a', b', c')$-fragility then $f(\bw)$ contains an $(a, b, c)$-fragility. 
For instance, if $A = \{a, b, c\}$ and 
$\bw$ contains an $(a, b, c)$-fragility $u$, then this fragility is propagated by the morphism $f = [a, ab, ac]$.
Note that $|f(v)\alpha| \leq  |u|$ implies $|v| < |u|$. Hence the propagation of a fragility makes it strictly longer.

Some fragilities can also be not propagated.
For instance, the morphism $[a, ab, abc]$ does not propagate the $(b,a,c)$-fragilities.

\begin{lemma}
\label{origine fragilities}
Let $\bw$ be an infinite word (not necessarily LSP) 
and let $f$ be an R-bLSP morphism defined over ${\rm alph}(\bw)$.
Fragilities of $f(\bw)$ are exactly the fragilities associated with $f$ and the fragilities of $\bw$ propagated by $f$.
\end{lemma}

\begin{proof}
From the definitions given before the lemma, the fragilities associated with $f$ and the fragilities of $\bw$ propagated by $f$ are fragilities of $f(\bw)$.
Let $u$ be an $(a, b, c, \beta, \gamma)$-fragility of $f(\bw)$.

%\textit{(New fragilities)}
If $u = \varepsilon$, 
it follows from the definition of an $(a, b, c, \beta, \gamma)$-fragility 
that $a = \first(f(\bw))$ and $\beta b$, $\gamma c$ are factors of $f(\bw)$.
Now observe that, still by the same definition, $a \not\in \{b, c\}$.
Thus by the definition of R-bLSP morphisms, $a = \first(f)$ and $\beta b$, $\gamma c$ belong to ${\rm Fact}(f({\rm alph}(\bw)))$.
The fragility $u$ is associated with $f$.

%\textit{(Propagated fragilities)}
Assume from now on that $u$ is not empty.
Let $\alpha = \first(f)$.
Considering the last occurrence of $\alpha$ in $u$, 
observe that the word $u$ can be decomposed in a unique way 
as $u = f(v) \alpha x$ with $v$, $x$ words such that $|x|_\alpha = 0$.
As $u$ is an $(a, b, c, \beta, \gamma)$-fragility of $f(\bw)$, there exist words $w_1$, $w_2$ and $w_3$ such that:
\begin{itemize}
\itemsep0cm
\item $|w_1|_\alpha = |w_2|_\alpha = |w_3|_\alpha = 0$;
\item $f(v) \alpha x w_1 \alpha$ is a prefix of $f(\bw)$ and $a = \first(w_1\alpha)$;
\item $\beta f(v)\alpha x w_2 \alpha$ and $\gamma f(v) \alpha x w_3 \alpha$ are factors of $f(\bw)$ with $b = \first(w_2\alpha)$ and $c = \first(w_3 \alpha)$.
\end{itemize}
By the definition of an R-bLSP morphism, there exist letters $a'$, $b'$, $c'$ such that $f(a') = \alpha x w_1$, $f(b') = \alpha x w_2$, $f(c') = \alpha x w_3$. 
These letters $a'$, $b'$, $c'$ are pairwise distinct since letters $a = \first(w_1\alpha)$, $b = \first(w_2\alpha)$ and $c = \first(w_3 \alpha)$ are pairwise distinct.
Moreover $va'$ is a prefix of $\bw$ and words $\beta v b'$ and $\gamma v c'$ are factors of $\bw$ 
(remember that $\alpha$ marks the beginning of letters in $f(\bw)$ as 
$f$ is an R-bLSP morphism and as, for letters $x$ and $y$, last$(f(x)) = y$ implies $x = y$). 
Hence the word $v$ is an $(a', b', c', \beta, \gamma)$-fragility of $\bw$.
This fragility is propagated by $f$.
\end{proof}

\section{\label{sec:automaton}A first automaton}

Lemma~\ref{origine fragilities} shows that bLSP morphisms act locally on fragilities on the words. 
One can then construct an automaton to store the actions of these morphisms on fragilities of LSP words 
(one can do this for arbitrary words but this is not needed for our purpose). 
It is important to note that the exact alphabet on which a morphism is applied has a lot of importance. 
For instance, if $\bw$ is an infinite word over $\{a, b, c\}$ and 
if $f = \lambda_{abc} = [a, ab, abc]$, 
$f(\bw)$ contains an $(a, b, c)$-fragility if and only if $c$ occurs in ${\rm alph}(\bw)$.
It is also important to note that the automaton will not follow the fragilities themselves but only the $3$-tuples of letters $(a, b, c)$ for which a fragility occurs.
Given an infinite word $\bw$, let $\frag(\bw)$ be the set of $3$-tuples $(a, b, c)$ of letters
such that $\bw$ contains an $(a, b, c)$-fragility.

Let ${\cal A}_1 = (\SRbLSP, Q_1, \Delta_1)$ be the automaton without initial and final states defined by:
\begin{itemize}
\item The alphabet of ${\cal A}_1$ is the set $\SRbLSP$ of R-bLSP morphisms over the alphabet $A$.

\item The set of states is the set $\{({\rm alph}(\bw), \frag(\bw)) \mid \bw$ LSP over $A \}$.

\item A $3$-tuple $((A_1, F_1), f, (A_2, F_2))$ is a transition of ${\cal A}_{1}$ if and only if the following condition holds:
	\begin{enumerate}
	\item $f_{|A_1} = f$ ($f$ is exactly defined on $A_1$; this point is important as it prevents from having a morphism in the transition that creates a fragility which is not related to $A_1$; we need to control new and propagated fragilities);
	\item $A_2 = {\rm alph}(f(A_1))$ (each letter of $A_2$ occurs in at least one image of a letter in $A_1$).
	\item if $(a, b, c) \in A_1$ then $f$ is not LSP $(a,b,c)$-breaking;
	\item $F_2$ is the union of the set of $3$-tuples $(a, b, c)$ for which an $(a, b, c)$-fragility  is associated with $f$ 
	   and the set of $3$-tuples $(a, b, c)$ propagated by $f$ from $(a', b', c')$-fragilities with $(a', b', c')$ in $F_1$.
	\end{enumerate}
\end{itemize}

One can observe that for each state $q$ of the automaton and each R-bLSP morphism $f$, there exists at most one state $q'$
such that $(q, f, q')$ is a transition.
There is no transition when $q = (B, F)$, $(a, b, c) \in F$ and $f$ is $(a, b, c)$-breaking.

By the definition of the automaton, Lemma~\ref{origine fragilities} has the following immediate corollary.

\begin{corollary}
\label{C:following fragilities}
Let $\bw$ be an LSP word, let $A_0, A_1, \ldots, A_k$ be alphabets and let $f_1$, $f_2$, \ldots, $f_k$ be R-bLSP morphisms such that:
\begin{itemize}
\item $A_0 = {\rm alph}(\bw)$;
\item for all $i$, $1 \leq i \leq k$, $A_i = {\rm alph}(f_i \circ \cdots \circ f_1(\bw))$;
\item For all $i$, $1 \leq i \leq k$, $f_i \circ \cdots \circ f_1(\bw)$ is LSP.
\end{itemize}
Then in ${\cal A}_1$ there is a unique path labeled by $f_1 \cdots f_k$ from the state $({\rm alph}(\bw), \frag(\bw))$ to a state $(A_k, F)$. Moreover $F$ is the set of $3$-tuples $(a, b, c)$ for which there exist $(a, b, c)$-fragilities in $f_k\circ \cdots \circ f_1(\bw)$.
\end{corollary}

On the alphabet $\{a, b\}$, as there is no fragility, the automaton ${\cal A}_1$ has only three states (see Figure~\ref{Graph_A1_binary}): one state for each of the alphabets $\{a\}$, $\{b\}$, $\{a, b\}$.

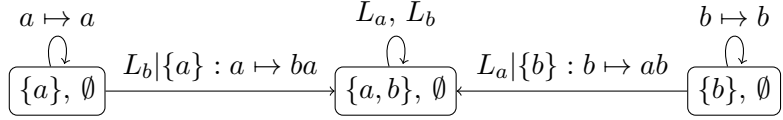
\begin{figure}[!ht]
\begin{center}
\begin{tikzpicture}[scale=0.3]
\begin{scope}
\node[draw, rectangle, rounded corners = 3pt] (a) at (0,0) {$\{a\}$, $\emptyset$};
\path  (a) edge [loop above] node[above]{$a \mapsto a$} (a) ;

\node[draw, rectangle, rounded corners = 3pt] (b) at (30,0) {$\{b\}$, $\emptyset$};
\path  (b) edge [loop above] node[above]{$b \mapsto b$} (b) ;

\node[draw, rectangle, rounded corners = 3pt] (ab) at (15,0) {$\{a, b\}$, $\emptyset$};
\path  (ab) edge [loop above] node[above]{$L_a$, $L_b$} (ab) ;
\draw [->] (a) -- (ab) node[midway,above]{$L_b|\{a\}: a \mapsto ba$};
\draw [->] (b) -- (ab) node[midway,above]{$L_a|\{b\}: b \mapsto ab$};

\end{scope}
\end{tikzpicture}
\end{center}
\caption{\label{Graph_A1_binary}${\cal A}_{1}$ for the binary alphabet}
\end{figure}

The automaton for the alphabet $\{a, b, c\}$ is provided in Figure~\ref{Graph_A1_ternary}.
Here follow some explanations.
We say that $\bw$ has fragilities of type $a$ (resp. of type $b$, of type $c$) 
if it has $(a, b, c)$-fragilities or $(a, c, b)$-fragilities (resp. $(b, a, c)$-fragilities or $(b, c, a)$-fragilities; $(c, a, b)$-fragilities or $(c, b, a)$-fragilities). 
Figure~\ref{table} shows the action of R-bLSP morphisms on fragilities.
One can observe that, for this ternary case, the type of a propagated fragility is the same then the type of the fragility from which it is propagated.

\begin{figure}[!ht]
\begin{center}
\begin{tabular}{|c|c|c|c|c|}
\hline
Properties of the & new fragilities & action on & action on & action on  \\
R-bLSP morphism $f$ &  &type $a$ fragilities  & type $b$ fragilities  & type $c$ fragilities\\
considered & & & & \\
\hline
$f$ is $L_a$, $L_b$, $L_c$ & & propagated & propagated & propagated \\
or one of their restrictions & &  &  &  \\
\hline
$f(c) = abc$ or $f(b) = acb$ & $a$-fragilities & LSP breaking & not propagated & not propagated \\
\hline
$f(c) = bac$ or $f(a) = bca$ & $b$-fragilities & not propagated & LSP breaking & not propagated \\
\hline
$f(b) = cab$ or $f(a) = cba$ & $c$-fragilities  & not propagated & not propagated & LSP breaking \\
\hline
\end{tabular}
\end{center}
\caption{\label{table}Ternary case: action of R-bLSP morphisms on fragilities}
\end{figure}

Observe that in the last three cases of Figure~\ref{table} only one type of fragility is kept. 
This shows that for ternary LSP words only four sets of fragilities can occur: the emptyset, the set of $a$-fragilities, the set of $b$-fragilities and the set of $c$-fragilities. In other words an LSP word over $\{a, b, c\}$ cannot have simultaneously 
fragilities of type $a$ and $b$ (nor of type $a$ and $c$; nor of type $b$ and $c$). 

The presentation of the automaton in Figure~\ref{Graph_A1_ternary} is split into two parts.
The first part contains the seven states corresponding to words without fragilities.
The second part contains the three states  corresponding to words with fragilities.
In the first part, morphisms with entries ``-" denote restrictions of endomorphisms: ``-" means ``not defined".
In the second part, entries ``*" denote any possibility of image (defined or not defined as this kind of entries correspond to restriction of $\lambda_{.}$ morphisms).
There exist numerous transitions from the first part to the second one. The outgoing side of the transition is presented in the first part and the in-going side of the transition is presented in the second part.
For instance there exists a transition 
from state $(\{b, c\}, \emptyset)$
to state $(\{a, b, c\}$, type $a)$ labeled by morphism $[-, ab, abc]$.

\begin{figure}[!ht]
\begin{center}
\begin{tikzpicture}[scale=0.3]
\begin{scope}
%\node (exc) at (0, 0) {ex-type c};
%\node (exa) at (6, 12) {ex-type a};
%\node (exb) at (12, 0) {ex-type b};

%% ETATS COMPOSANTE 1
\node[draw, rectangle, rounded corners = 3pt] (aa) at (12,0) {\tiny $\{a\}$, $\emptyset$};
\node[draw, rectangle, rounded corners = 3pt] (bb) at (6,10) {\tiny $\{b\}$, $\emptyset$};
\node[draw, rectangle, rounded corners = 3pt] (cc) at (18,10) {\tiny $\{c\}$, $\emptyset$};

\node[draw, rectangle, rounded corners = 3pt] (ab) at (0,0){\tiny $\{a, b\}$, $\emptyset$};
\node[draw, rectangle, rounded corners = 3pt] (ac) at (24,0) {\tiny $\{a, c\}$, $\emptyset$};
\node[draw, rectangle, rounded corners = 3pt] (bc) at  (12,16)  {\tiny $\{b, c\}$, $\emptyset$};

\node[draw, rectangle, rounded corners = 3pt] (d) at (12, 8) {\tiny $\{a, b, c\}$, $\emptyset$};%(12, -12)

%% LOOPS SUR ETATS COMPOSANTE 1
\path  (aa) edge [loop below,>=latex] node[pos=0.5]{\tiny $[a, -, -]$} (aa) ;
\path  (bb) edge [loop left,>=latex] node[pos=0.5]{\tiny $[-, b, -]$} (bb) ;
\path  (cc) edge [loop right,>=latex] node[pos=0.5]{\tiny $[-, -, c]$} (cc) ;
\path  (ab) edge [loop left,>=latex] node[pos=0.35]{\tiny $[a, ab, -]$} (ab) ;
\path  (ab) edge [loop left,>=latex]  node[pos=0.65]{\tiny $[ba, b, -]$} (ab) ;
\path  (ac) edge [loop right,>=latex] node[pos=0.35]{\tiny $[a, -, ac]$} (ac) ;
\path  (ac) edge [loop right,>=latex]  node[pos=0.65]{\tiny $[ca, -, c]$} (ac) ;
\path  (bc) edge [loop above,>=latex] node[pos=0.5]{\tiny $[-, b, bc]$, $[-, cb, c]$} (bc) ;

% Autre transition composante 1
\draw [->,>=latex] (ab) -- (d) node[pos=0.4]{\tiny $[ca, cb, -]~~$};
\draw [->,>=latex] (ac) -- (d) node[pos=0.4]{\tiny $~~[ba, -, bc]$};
\draw [->,>=latex] (bc) to node[pos=0.4]{\tiny $[-, ab, ac]$} (d);

\draw [->,>=latex] (aa)  to node[pos=0.5, below]{\tiny $~[ba, -, -]$} (ab);%[bend left]
\draw [->,>=latex] (aa)  to node[pos=0.5, below]{\tiny $[ca, -, -]~$} (ac) ;%[bend right]
\draw [->,>=latex] (bb)  to node[pos=0.3] {\tiny $[-, ba, -]$} (ab) ;%[bend right]
\draw [->,>=latex] (bb) to node[pos=0.2]{\tiny $~~[-, bc, -]$} (bc) ;%[bend left]
\draw [->,>=latex] (cc) to node[pos=0.5]{\tiny $[-, -, ac]$} (ac) ;%[bend left]
\draw [->,>=latex] (cc) to node[pos=0.15]{\tiny $[-, -, bc]~~$} (bc) ;%[bend right]

% COMPOSANTE 1 vers COMPOSANTE 2
\node (bc_vers_a) at (20,16) {\tiny type $a$};
\node (bc_vers_b) at (4,18) {\tiny type $b$};
\node (bc_vers_c) at (4,14) {\tiny type $c$};
\draw [->,>=latex] (bc) to node[above]{\tiny $[-,ab,abc]$} node[below]{\tiny $[-, acb, ac]$} (bc_vers_a) ;

\draw [->,>=latex] (bc) to node[pos=0.5,below]{\tiny $[-, b, bac]$}  (bc_vers_b) ;
\draw [->,>=latex] (bc) to node[pos=0.5, below]{\tiny $[-, cab, c]$}  (bc_vers_c) ;

\node (ab_vers_a) at (-6,-6) {\tiny type $a$};
\node (ab_vers_b) at (0,-6) {\tiny type $b$};
\node (ab_vers_c) at (6,-6) {\tiny type $c$};
\draw [->,>=latex] (ab) to node[pos=0.2]{\tiny $[a, acb,-]$} (ab_vers_a) ;
\draw [->,>=latex] (ab) to node[pos=0.4]{\tiny $[bca, b, -]$} (ab_vers_b) ;
\draw [->,>=latex] (ab) to node[pos=0.6]{\tiny $[ca, cab, -]$} node[pos=0.8]{\tiny $[cba, cb, -]$} (ab_vers_c) ;

\node (ac_vers_a) at (18,-6) {\tiny type $a$};
\node (ac_vers_b) at (24,-6) {\tiny type $b$};
\node (ac_vers_c) at (30,-6) {\tiny type $c$};
\draw [->,>=latex] (ac) to node[pos=0.2]{\tiny $[a, -, abc]$} (ac_vers_a) ;
\draw [->,>=latex] (ac) to node[pos=0.4]{\tiny $[ba, -, bca]$} node[pos=0.6]{\tiny $[bca, -, bc]$} (ac_vers_b) ;
\draw [->,>=latex] (ac) to  node[pos=0.8]{\tiny $[cba, -, c]$} (ac_vers_c) ;

\node (d_vers_a) at (6,2) {\tiny type $a$};
\node (d_vers_b) at (12,2) {\tiny type $b$};
\node (d_vers_c) at (18,2) {\tiny type $c$};
\draw [->,>=latex] (d) to node[pos=0.5]{\tiny $\lambda_{abc}$} node[pos=0.7]{\tiny $\lambda_{acb}$} (d_vers_a) ;
\draw [->,>=latex] (d) to node[pos=0.5]{\tiny $\lambda_{bac}$} node[pos=0.7]{\tiny $\lambda_{bca}$} (d_vers_b) ;
\draw [->,>=latex] (d) to  node[pos=0.5]{\tiny $\lambda_{cab}$} node[pos=0.7]{\tiny $\lambda_{cba}$} (d_vers_c) ;

\node (bb_vers_a) at (0,13) {\tiny type $a$};
\node (bb_vers_c) at (0,7) {\tiny type $c$};
\draw [->,>=latex] (bb) to node[pos=0.5]{\tiny $[-, acb, -]$} (bb_vers_a) ;
\draw [->,>=latex] (bb) to node[pos=0.5]{\tiny $[-, cab, -]$} (bb_vers_c) ;

\node (cc_vers_a) at (23,13) {\tiny type $a$};
\node (cc_vers_b) at (23,7) {\tiny type $b$};
\draw [->,>=latex] (cc) to node[pos=0.5]{\tiny $[-, -, abc]$} (cc_vers_a) ;
\draw [->,>=latex] (cc) to node[pos=0.5,right]{\tiny $[-, -, bac]$} (cc_vers_b) ;

\node (aa_vers_b) at (7,-5) {\tiny type $a$};
\node (aa_vers_c) at (17,-5) {\tiny type $b$};
\draw [->,>=latex] (aa) to node[pos=0.7]{\tiny $[bac, -, -]$} (aa_vers_b) ;
\draw [->,>=latex] (aa) to node[pos=0.5,right]{\tiny $[cba, -, -]$} (aa_vers_c) ;

%% ETATS COMPOSANTE 2

\node[draw, rectangle, rounded corners = 3pt] (a) at (12,-19) {\tiny $\{a, b, c\}$, type $a$};
\node[draw, rectangle, rounded corners = 3pt] (b) at (18,-11) {\tiny $\{a, b, c\}$, type $b$};
\node[draw, rectangle, rounded corners = 3pt] (c) at (6,-11) {\tiny $\{a, b, c\}$, type $c$};

% Transitions composante 2
\draw [->,>=latex] (c.275) to[bend right] node[pos=0.5,left]{\tiny $\lambda_{abc}$} (a.110) ;
\draw [->,>=latex] (c.275) to[bend right] node[pos=0.6,left]{\tiny $\lambda_{acb}$} (a.110) ;
\draw [->,>=latex] (a.110) to[bend right] node[pos=0.4, left]{\tiny $\lambda_{cab}$} (c.275) ;
\draw [->,>=latex] (a.110) to[bend right] node[pos=0.5, left]{\tiny $\lambda_{cba}$} (c.275) ;

\draw [->,>=latex] (a.70) to[bend right] node[pos=0.4,right]{\tiny $\lambda_{bca}$} (b.265);
\draw [->,>=latex] (a.70) to[bend right] node[pos=0.5, right]{\tiny $\lambda_{bac}$} (b.265) ;
\draw [->,>=latex] (b.265) to[bend right] node[pos=0.5, right]{\tiny $\lambda_{acb}$} (a.70) ;
\draw [->,>=latex] (b.265) to[bend right] node[pos=0.6, right]{\tiny $\lambda_{abc}$} (a.70) ;

\draw [->,>=latex] (c.335) to[bend right] node[pos=0.5, above]{\tiny $~~\lambda_{bca}$, $\lambda_{bac}$} (b.205) ;
\draw [->,>=latex] (b.155) to[bend right] node[pos=0.5, above]{\tiny $\lambda_{cab}$, $\lambda_{cba}$} (c.25) ;

\path  (c) edge [loop left,>=latex] node[pos=0.3]{\tiny $L_a$} (c) ;
\path  (c) edge [loop left,>=latex] node[pos=0.5]{\tiny $L_b$} (c) ;
\path  (c) edge [loop left,>=latex] node[pos=0.7]{\tiny $L_c$} (c) ;

\path  (b) edge [loop right,>=latex] node[pos=0.3]{\tiny $L_a$} (b) ;
\path  (b) edge [loop right,>=latex] node[pos=0.5]{\tiny $L_b$} (b) ;
\path  (b) edge [loop right,>=latex] node[pos=0.7]{\tiny $L_c$} (b) ;

\path  (d) edge [loop right,>=latex] node[pos=0.25]{\tiny $L_a$} (d) ;
\path  (d) edge [loop right,>=latex] node[pos=0.5]{\tiny $L_b$} (d) ;
\path  (d) edge [loop right,>=latex] node[pos=0.7]{\tiny $L_c$} (d) ;

\path  (a) edge [loop right,>=latex] node[pos=0.3]{\tiny $L_a$} (a) ;
\path  (a) edge [loop right,>=latex] node[pos=0.5]{\tiny $L_b$} (a) ;
\path  (a) edge [loop right,>=latex] node[pos=0.7]{\tiny $L_c$~} (a) ;

% Entrees dans composantes 2
\node (vers_a) at (12,-23) {};
\node (vers_b) at (18,-7) {};
\node (vers_c) at (6,-7) {};

\draw [->,>=latex] (vers_a) to node[pos=0.5]{\tiny $[*, *, abc]$} node[pos=0.3]{\tiny $[*, acb *]$} (a) ;
\draw [->,>=latex] (vers_b) to node[pos=0.5]{\tiny $[*, *, bac]$} node[pos=0.3]{\tiny $[bca, *, *]$}  (b) ;
\draw [->,>=latex] (vers_c) to node[pos=0.5]{\tiny $[*, cab, *]$} node[pos=0.3]{\tiny $[cba, *, *]$} (c) ;

\end{scope}
\end{tikzpicture}
\end{center}
\caption{\label{Graph_A1_ternary}${\cal A}_{1}$ for the ternary alphabet}
\end{figure}

For larger alphabets, the number of states explodes. 
This is due to the increasing number of subalphabets 
but also to the fact that the set of fragilities may contain more than one fragility. 
For instance, any LSP morphism $f$ such that $f(a) = abcd$ creates simultaneously several fragilities: in the image of an infinite word containing $a$, the empty word is an $(a, b, c)$-fragility, an $(a, b, d)$-fragility and an $(a, c, d)$-fragility.

\section{\label{sec:carac}The {\it S}-adic Characterization of LSP words}

Let $\bw$ be an LSP word. 
By Proposition~\ref{characLSPwords}, $\bw$ is $\SRbLSP$-adic.
Let $\mbf = (f_k)_{k \geq 0}$ be its fitted directive word over $\SRbLSP$.
Let also $(\bw_k)_{k \geq 1}$ be the sequence of infinite words verifying: $\bw_1 = \bw$; 
for all $k \geq 1$, $\bw_k = f_k(\bw_{k+1})$.

Let ${\cal A}_{R-bLSP}(A)$, or simply ${\cal A}_{R-bLSP}$, be the automaton obtained by reversing transitions in ${\cal A}_1$
and considering all states as initial.
As $\mbf$ is fitted,
an immediate consequence of Corollary~\ref{C:following fragilities} is 
that there exists an infinite path labeled by $\mbf$ in ${\cal A}_{R-bLSP}$ starting on state $({\rm alph}(\bw), \frag(\bw))$.
The next theorem states that conversely if such an infinite path exists for
an $\SRbLSP$-adic word $\bw$, then $\bw$ is LSP.
If $q = (X, F)$ is a state of ${\cal A}_{R-bLSP}$, let ${\rm alph}(q) = X$ and $\frag(q) = F$.

\begin{theorem}
\label{th_carac}
A word $\bw$ is LSP if and only if it is $\SRbLSP$-adic and 
there exists an infinite path in ${\cal A}_{R-bLSP}$
labeled by the fitted directive word of $\bw$.
\end{theorem}

Applying this theorem allows to provide particular examples of LSP words.
For instance it can be verified that all standard episturmian words are LSP \cite{ArnouxRauzy1991,JustinPirillo2002}. These words are the ${\cal S}$-adic words where $S$ is the set of morphisms containing morphisms $L_\alpha$ and their restrictions. 
In particular, if $A = \{a_1, \ldots, a_k\}$, the $k$-bonacci word $\bB_{a_1, \ldots, a_k}$ (or Tribonacci word, also denoted $\bT$, when $k = 3$) which is the fixed point of the morphism $L_{a_1}\circ \cdots \circ L_{a_k}$ is LSP. 
Observe that $\bB_{a_1, \ldots, a_k}$ contains no fragility as
its fitted directive word is recognized using only the state 
$(\{a_1, \ldots, a_k\}, \emptyset)$ in ${\cal A}_{R-bLSP}$.
It may be observed that the fitted directive word of an LSP word can be recognized by several paths.
For instance  the fitted directive word of the Tribonacci word is the word $(L_aL_bL_c)^\omega$ which is recognized by four paths starting respectively from states $(\{a, b, c\}, \emptyset)$, $(\{a, b, c\}, {\rm type~a})$,  $(\{a, b, c\}, {\rm type~b})$ and  $(\{a, b, c\}, {\rm type~c})$.

%It is important to observe that ${\cal A}_{R-bLSP}$ does not recognize all directive words of LSP words. 
%For instance, the LSP word $(abc)^\omega$ has $[-, -, abc][-, -, c]^\omega$ as fitted directive word. 
%The word $\lambda_{abc}[-, -, c]^\omega$ is another directive word of $(abc)^\omega$ but it is not recognized by ${\cal A}_{R-bLSP}$.

\begin{proof}[Proof of Theorem~\ref{th_carac}]
The only if part was already mentioned before the statement of the theorem.

Assume, by contradiction, that there exists in ${\cal A}_{R-bLSP}$ a path labeled by a directive word $\mbf$ of a word $\bw$
which is $\SRbLSP$-adic but not LSP
(for this part of the proof it is not needed that $\mbf$
is a fitted directive word).
Such a word contains a left special factor $u$ that is not a prefix of $\bw$.
Among all possible $3$-tuples $(\mbf, \bw, u)$, choose one such that $|u|$ is minimal.

For $n \geq 1$, we denote by $f_n$ the $n^{th}$ letter of $\mbf$ and $\bw_n$ the word directed by $(f_k)_{k\geq n}$ ($\bw_1 = \bw$ ; $\bw_2$ is the word directed by 
$f_2 f_3 \cdots$; $\bw_n = f_n( \bw_{n+1})$ for all $n \geq 1$).

The first three steps do not depend on the automaton.

\medskip

\textit{Step 1: $\bw_2$ contains a fragility} 

First observe that $|u| \geq 2$. 
Indeed we have $|u|\neq 0$ as the empty word is a prefix of $\bw$.
Moreover, by the structure of images of the R-bLSP morphism $f_1$,
only the letter $\first(f_1)$ can be left special, whence $|u| \neq 1$.

Let $\alpha = \first(\bw) = \first(f_1)$.
Considering the last occurrence of $\alpha$ in $u$,
the word $u$ can be decomposed in a unique way $u = f_1(v) \alpha x$ with $v$, $x$ words such that $|x|_\alpha = 0$.

As $u$ is left special, there exist distinct letters $\beta$ and $\gamma$ 
such that $\beta u$ and $\gamma u$ are factors of $\bw$.
As the letter $\alpha$ marks the beginning of images of letters in $\bw$
and as for all letters $\delta$, $f_1(\delta)$ ends with $\delta$,
we deduce that 
$\beta v $ and $\gamma v$ are factors of $\bw_2$.
As $|v| < |u|$ and by the choice of the $3$-tuple $(\mbf, \bw, u)$,  the word $v$ is a prefix of $\bw_2$.
Consequently $f_1(v)\alpha$ is a prefix of $\bw$ and so $x \neq \varepsilon$.

Assume there exists a unique letter $b$ such that 
$\beta vb$ is a factor of $\bw_2$ and $u$ is a prefix of $f(vb)$.
Assume also that $b$ is the unique letter $c$ such that $\gamma vc$ 
is a factor of $\bw_2$ and $u$ is a prefix of $f(vc)$.
As $u$ is not a prefix of $\bw = f_1(\bw_2)$ and as $u$ is a prefix of $f_1(vb)$,
the word $vb$ is not a prefix of $\bw_2$.
By the choice of the $3$-tuple $(\mbf, \bw, u)$, $|vb| \geq |u|$.
As $|v|<|u|$, we get $|vb| = |u| = |f_1(v)\alpha x|$.
As $|f_1(v)| \geq |v|$, it follows $x = \varepsilon$: a contradiction.

From what precedes, we deduce the existence of two distinct letters
$b$ and $c$ such that $\beta v b$ and $\gamma v c$ are factors of $\bw_2$ 
with $u$ a prefix of $f_1(vb)$ and $f_1(vc)$. As $u$ is not a prefix of $\bw = f_1(\bw_2)$, 
the letter $a$ that follows the prefix $v$ of $\bw_2$ is different from $b$ and $c$.
Hence the word $\bw_2$ is $(a, b, c, \beta, \gamma)$-fragile and $v$ is such a fragility.

\medskip

\textit{Step 2: $f_1$ is LSP $(a, b, c)$-breaking}

By the definition of letters $b$ and $c$ at Step~1, the word $\alpha x$ is a common prefix of $f_1(b)$ and $f_1(c)$. Also as $u = f_1(v)\alpha x$ is not a prefix of $\bw$ while $f_1(v)a$ is a prefix of $\bw$,
the word $\alpha x$ is not a prefix of $f_1(a)$.
By Corollary~\ref{carac LSP breaking morphisms},
$f_1$ is $(a, b, c)$-breaking.

\medskip

\textit{Step 3: origin of fragilities of $\bw_2$}

Applying iteratively Lemma~\ref{origine fragilities},
as the lengths of propagated fragilities decrease with the propagation, 
we deduce the existence of an integer $n \geq 2$ and
a sequence of $3$-tuples of pairwise different letters
$(a_i, b_i, c_i)_{i \in \{2, \cdots, n\}}$, 
a sequence $(v_i)_{i \in \{2, \cdots, n\}}$ of words such that:
\begin{itemize}
\item $v_i$ is an $(a_i, b_i, c_i, \beta, \gamma)$-fragility of $\bw$ 
for all $i \in \{2, \cdots, n\}$;
\item $(a_2, b_2, c_2) = (a, b, c)$ and $v_2 = v$;
\item $|v_{i+1}| < |v_i|$ for all $i \in \{2, \cdots, n-1\}$  ($v_i$ is a fragility propagated by $f_i$ from $v_{i+1}$);
\item $v_n = \varepsilon$ (origin of the fragilities);
\item the words $v_i a_i$, $\beta v_i b_i$, $\gamma v_i c_i$ are respectively prefixes of 
$f_i( v_{i+1} a_{i+1})\alpha_i$,
$\beta f_i( v_{i+1} b_{i+1})\alpha_i$,
$\gamma f_i( v_{i+1} c_{i+1})\alpha_i$
where $\alpha_i = \first(f_i)$ for $i \in \{2, \cdots, n-1\}$ (by the definition of propagated fragilities);
\item $a_n = \first(f_n)$;
\item $\beta b_n$, $\gamma c_n$ belong to $f_n({\rm alph}(\bw_{n+1}))$  (by the definition of new fragilities).
\end{itemize}

\medskip

\textit{Step 4: conclusion}.

Let $(q_i)_{i \geq 1}$ be the sequence of states along the path labeled by $\mbf$: 
for all $n \geq 1$, $(q_n, f_n, q_{n+1})$ is a transition of ${\cal A}_{R-bLSP}$.
At the end of Step 3, we learn that there exists an $(a_n, b_n, c_n)$-fragility in $f_n(\bw_{n+1})$. 
Hence
$a_n$, $b_n$, $c_n$ are pairwise distinct letters. 
Especially as $a_n = \first(f_n) \not\in \{b_n, c_n\}$ by properties of R-bLSP morphisms,
the words $\beta b_n$ and $\gamma c_n$ are factors of images of some letters, say $b_n'$ and $c_n'$.
By the definition of the automaton ${\cal A}_{R-bLSP}$,
${\rm alph}(q_n) = {\rm alph}(f_n({\rm alph}(q_{n+1})))$ 
and $f_n$ is defined exactly on ${\rm alph}(q_{n+1})$.
This implies that $b_n'$ and $c_n'$ belong to ${\rm alph}(q_{n+1})$ and $a_n$, $b_n$ and $c_n$ belong to ${\rm alph}(q_n)$. Moreover, as $\beta b_n$, $\gamma c_n$ are factors of words in $f_n({\rm alph}(q_{n+1}))$, we deduce that $(a_n, b_n, c_n) \in \frag(q_n)$.

Using a backward induction and the definition of the automaton ${\cal A}_{R-bLSP}$, 
we can show that for all $i$, $2 \leq i \leq n$, 
$(a_i, b_i, c_i) \in \frag(q_i)$.
Especially $(a_2, b_2, c_2) \in \frag(q_2)$.
As $(q_1, f_1, q_2)$ is a transition of ${\cal A}_{R-bLSP}$, 
by the definition of transitions, $f_1$ is not LSP $(a_2, b_2, c_2)$-breaking: a contradiction with Step~2 as 
$(a, b, c)= (a_2, b_2, c_2)$.
\end{proof}

\section{\label{sec:preservation}Morphisms preserving LSP words}

As explained at the end of \cite{Richomme2017DLT}, one can ask for a simpler characterization for LSP words.
In the context of substitutive-adicity, the next theorem shows that this is not possible: given any set $S$ of morphisms such that LSP words are $S$-adic, one needs a way to distinguish directive words of LSP words.

\begin{theorem}
\label{T:no simpler result}
Let $A$ be an alphabet containing at least three letters.
There exists no set $S$ of morphisms such that the set of LSP infinite words over $A$ is the set of $S$-adic words.
\end{theorem}

The main idea of the proof lies on the fact that, given any set $S$ of morphisms, these morphisms preserve
the property of being an $S$-adic word: the image of any $S$-adic word by an element of $S$ is still $S$-adic.
Thus in this section, we first characterize the set of endomorphisms preserving the LSP property (the image of any LSP word is also LSP) and 
then show the existence of an LSP word that cannot be decomposed over this set of morphisms.

\subsection{Some Morphisms}

Let ${\cal L}(A) = \{ L_a \mid a \in A\}$. 
When the context will be clear, we will just write ${\cal L}$ instead of ${\cal L}(A)$. 
Observe that, in the binary case, ${\cal S}_{bLSP} = {\cal L}$.

\begin{lemma}
\label{L:La1}
A bLSP morphism preserves the LSP property if and only if it belongs to ${\cal L}$.
\end{lemma}

\begin{proof}
Let $f$ be a bLSP morphism that does not belong to ${\cal L}$.
There exist three pairwise distinct letters $a$, $b$ and $c$ such that
$f(a) = a$, $f(b) = ab$, $f(c) = abc$. 
As previously mentioned (see Example~\ref{Exemple3})
there exists an LSP word $\bw$ that contains an $(a, b, c)$-fragility (see the beginning of Section~\ref{sec:fragility}).
That is, for a word $u$ and distinct letters $\beta$, $\gamma$, $ua$ is a prefix of $\bw$ while $\beta u b$ and $\gamma uc$ are factors of $f(\bw)$.
Then $f(u)aa$ is a prefix of $f(\bw)$ while
$\beta f(u)ab$ and $\gamma f(u)ab$ are factors of $\bw$: $f(\bw)$ is not LSP and so $f$ does not preserve the LSP property.

Conversely for any LSP word $\bw$ and any letter $\alpha$, let us prove that $L_\alpha(\bw)$ is also LSP.
Assume by contradiction that $L_\alpha(\bw)$ is not LSP and consider a word $u$ and two distinct letters $a$ and $b$, such that $ub$ is a left special factor of $L_\alpha(\bw)$ while $ua$ is a prefix of $L_\alpha(\bw)$. 
Let $\beta$ and $\gamma$ be distinct
 letters such that $\beta ub$ and $\gamma ub$ are factors of $L_\alpha(\bw)$.
Observe that at least one of the letters $a$ and $b$ is different from $\alpha$, and also at least one of the letters 
$\beta$ and $\gamma$ is different from $\alpha$. 
The structure of $L_\alpha$ (any letter different from $\alpha$ is preceded by $\alpha$ in the image of any word)
implies that $u$ ends and begins with the letter $\alpha$ (in particular it is not empty) and so can be decomposed $u = L_\alpha(v)\alpha$ for a word $v$.
One can deduce that $va$ is a prefix of $\bw$ while $\beta v b$ and $\gamma v b$ are factors of $\bw$:  a contradiction with
$\bw$ LSP.
\end{proof}

Observe that Lemma~\ref{L:La1} does not extend to R-bLSP morphisms. 
For instance $[-, -, abc, ad]$ preserves the LSP property for words over $\{c, d\}$.

It is important also to observe that the first part of the proof of Lemma~\ref{L:La1} also shows the next result.

\begin{lemma}
\label{L:La1bis}
If a bLSP morphism does not belong to ${\cal L}$ then it is $(a, b, c)$-breaking for some pairwise distinct letters $a$, $b$ and $c$.
\end{lemma}

We will also need the following result.

\medskip %Extra space added in order to make "Lemma 19" appears

\begin{lemma}
\label{L:La2}
Let $\alpha$ be a letter, $a$, $b$ and $c$ be three distinct letters and $\bw$ be a word.
The word $\bw$ contains an $(a, b, c)$-fragility if and only if
$L_\alpha(\bw)$ contains an $(a, b, c)$-fragility.
\end{lemma}

\begin{proof}
The proof of the if part is similar to the converse part of the proof of Lemma~\ref{L:La1}
(if $u$ is a word such that $ua$ is a prefix of $L_\alpha(\bw)$ 
and the words $\beta u b$ and $\gamma uc$ are factors of $L_\alpha(\bw)$ with $\beta \neq \gamma$
then $u = L_\alpha(v)\alpha$ for a word $v$ such that $va$ is a prefix of $\bw$ and the words
$\beta w b$ and $\gamma w c$ are factors of $\bw$).
The proof of the only if part is straightforward and is left to readers.
\end{proof}

There exist morphisms that preserve the LSP property without being bLSP morphisms. We introduce some of them.
First, we denote by ${\cal U}(A, B)$ (or simply ${\cal U}$ when the context is clear), the set of all nonerasing morphisms 
from $A^*$ to $B^*$ 
such that for each letter $a$ in $B$, $\sum_{\alpha \in A} |f(\alpha)|_a \leq 1$ 
(each letter of $B$ occurs at most once in the set of images of letters by $f$). Here ${\cal U}$ stands for uniqueness. 
\textit{Renaming morphisms} are nonerasing morphisms such that, for all $a$ in $A$, $|f(a)| =1$ and such that
$\#\{f(a) \mid a \in A\} = A$. 
Renaming morphisms are elements of ${\cal U}(A, B)$. 
\textit{Permutation morphisms} (or simply \textit{permutations}), that is, morphisms such that $\{f(a) \mid a \in A\} = A$, are particular renaming morphisms. Let $Perm(A)$ (or simply $Perm$) denote the set of permutation morphisms.

Second, we denote by ${\cal P}_{\rm LSP}(A, B)$ (or simply ${\cal P}_{\rm LSP}$ when the context is clear), the set of all nonerasing morphisms 
from $A^*$ to $B^*$ for which there exists a word $u$ over $B$ with $f(\alpha) \subseteq u^+$ for all letters $\alpha$ in $A$ and
$u^\omega$ an LSP infinite word. Here ${\cal P}$ stands for periodic as any morphism in ${\cal P}_{\rm LSP}$ maps any infinite word to a periodic word. The proof of the next lemma follows quite immediately from the definitions and is left to the reader.

\begin{lemma}
All morphisms in ${\cal U}(A, B) \cup {\cal P}_{\rm LSP}(A, B)$ preserve the LSP property.
\end{lemma}

\subsection{A first necessary condition}

\begin{proposition}
\label{P:CN_preserving}
Any nonerasing morphism $f$ that preserves the LSP property for infinite words belongs to $\SRbLSP(B)^*{\cal U}(A, B) \cup {\cal P}_{LSP}(A, B)$.
\end{proposition}

The converse of this proposition is false as some R-bLSP morphisms do not preserve the LSP property (see the notion of LSP breaking morphisms).

The rest of the section is devoted to the proof of Proposition~\ref{P:CN_preserving}.
We first need a technical lemma on factors of LSP infinite words.
\begin{lemma}
\label{lemma_x3}
Let $\bw$ be an LSP word, $u$ be a prefix of $\bw$ and $n \geq 1$ be an integer. 
If $u^{n+1}$ is a factor of $\bw$, then $u^n$ is a prefix of $\bw$.
\end{lemma}

\begin{proof}
Let us consider an occurrence of the word $u^{n+1}$.
Observing its extension on the left, we deduce that $u = u_1u_2$ for two words $u_1$ and $u_2$ such that one of the following two possibilities hold:
\begin{itemize}
\item $u_2(u_1u_2)^{n+1}$ is a prefix of $\bw$,
\item $\alpha u_2(u_1u_2)^{n+1}$ is a factor of $\bw$ with $\alpha$ a letter different from $\last(u_2u_1)$.
\end{itemize}
In the second case, the word $(u_2u_1)^n$ is a left special factor of $\bw$.
As $\bw$ is LSP, $(u_2u_1)^n$ is a prefix of $\bw$.
This also holds in the first case.
By hypothesis, $u$ is a prefix of $\bw$. As $n \geq 1$, both $u$ and $u_2u_1$ are prefixes of $w$: $u = u_2u_1$ and consequently $u^n$ is a prefix of $\bw$. 
\end{proof}

For a non-empty finite word $w$, let $\sqrt{w}$ denote its primitive root, that is, the shortest word $u$ such that $w = u^k$ for some integer $k$.
We now show when  the second condition of Proposition~\ref{P:CN_preserving} happens.

\begin{lemma}
\label{lemma_x4}
Let $f$ be a nonerasing morphism that preserves the LSP property for infinite words.
Assume there exist two distinct letters $a$ and $b$ such that $\sqrt{f(a)} = \sqrt{f(b)}$. 
Then for all letters $\alpha$, $f(\alpha) \in \sqrt{f(a)}^+$, that is, $f \in {\cal P}_{\rm LSP}$.
\end{lemma}

\begin{proof}
Assume by contradiction that there exists a third letter $c$ such that 
$\sqrt{f(c)} \neq \sqrt{f(a)} = \sqrt{f(b)}$. 
Let $u = \sqrt{f(a)}$.
Consider the Tribonacci
word $\bT$ over $\{a, b, c\}$. By induction, one can prove that, for all $n \geq 1$, $L_a^{n-1}(\bT)$ 
begins with the prefix
$a^nba^nca$ and contains the word $a^nba^{n+1}ba^n$.
As $L_a^{n-1}(\bT)$ is an standard episturmian word, it is LSP.

Choose an integer $n$ such that $n \geq 1$
and $|f(a^nba^nc)u| \leq |f(a^nba^{n+1}ba^n)|-|u|$.
Let $m, p$ be the integers such that $f(a^nba^n) = u^p$ 
and $f(a^nba^{n+1}ba^n) = u^{m+1}$.
As $f$ preserves the LSP property, $f(L_a^{n-1}(\bT))$ is LSP.
As $u^{m+1}$ is a factor of $f(L_a^{n-1}(\bT))$ and $u$ is a prefix of $f(L_a^{n-1}(\bT))$, 
by Lemma~\ref{lemma_x3}, $u^m$ is a prefix of $f(L_a^{n-1}(\bT))$.
Also $u^pf(c)u$ is a prefix of $f(L_a^{n-1}(\bT))$
and $|u^pf(c)u| \leq |u^{m}|$.
Thus $f(c)$ is a power of $u$ or the last occurrence of $u$ in $u^pf(c)u$ is an internal factor of $u^2$ (i.e. $uu = p u s$ with $p$ and $s$ non-empty).
The former case is not possible as by hypothesis $\sqrt{f(c)} \neq u$.
The latter case contradicts the primitivity of $u$ (as it implies $u = ps = sp$, a well-known equation which implies that $u$ is not primitive -- see, \textit{e.g.}, \cite{Lothaire1983book}).
\end{proof}

Lemma~\ref{lemma_x4} implies that Proposition~\ref{P:CN_preserving} is a corollary of the next result.

\begin{proposition}
\label{P:CN_preserving2}
Let $f: A^* \to B^*$ be a nonerasing morphism that preserves the LSP property for infinite words.
Assume that for all letters $a$ and $b$, the words $f(a)$ and $f(b)$ do not have the same primitive root. 
Then $f \in \SRbLSP^* {\cal U}$.
\end{proposition}

\begin{proof}
We proceed by steps. Each step provides a stronger property on $f$.

\medskip

\noindent
\textbf{Step 1}. For all distinct letters $a$ and $b$ in $A$, 
$f(a)$ is not a suffix of $f(b)$.

\begin{proof}
Assume by contradiction that $f(a)$ is a suffix of $f(b)$ for two distinct letters $a$ and $b$.
Let $\ell \geq 0$ be the greatest integer such that $f(a)^\ell$ is a prefix of $f(b)$.
The word $f(a^{3+\ell})$ is a factor of $f(baab)$.
Let $\bw$ be any LSP word containing the factor $baab$ and having $aba$ as a prefix (for instance the Fibonacci word).
By Lemma~\ref{lemma_x3}, $f(a)^{2+\ell}$ is a prefix of $f(\bw)$.
Also $f(\bw)$ begins with $f(aba)$.
By the definition of $\ell$, this is possible only if 
$f(b) = f(a)^\ell u_1$,
$f(a) = u_1 u_2$ for words $u_1$ and $u_2$ with $u_2 \neq \varepsilon$ a prefix of $f(a)$.
But then, as $f(a)$ is a suffix of $f(b)$,
we have $\ell \geq 1$ and $f(a) = u_2 u_1$. Hence $u_1u_2 = u_2 u_1$.
Thus $u_1 = \varepsilon$ or the words $u_1$ and $u_2$ have the same primitive root by \cite[Prop. 1.3.2]{Lothaire1983book}.
It follows that $f(a)$ and $f(b)$ have also the same primitive root: a contradiction with our hypotheses.
\end{proof}

\noindent
\textbf{Step 2}. For all distinct letters $a$ and $b$ in $A$, 
$\last(f(a)) \neq \last(f(b))$.

\begin{proof}
Assume by contradiction that $\last(f(a)) = \last(f(b))$ for two distinct letters $a$ and $b$.
By Step~1, $f(a)$ is not a suffix of $f(b)$ and $f(b)$ is not a suffix of $f(a)$.
Hence there exist words $u$, $u_1$, $u_2$ and distinct letters $\alpha$, $\beta$ such that 
$f(a) = u_1\alpha u$, $f(b) = u_2 \beta u$ with $u \neq \varepsilon$.

We observe that for any LSP word $\bw$ containing $a$ and $b$, the word $u$ is a prefix of $f(\bw)$ 
(remember that $f$ preserves the LSP property, hence $f(\bw)$ is LSP and its left special factors are prefixes of it).
As $\bw$ can be chosen with $a$ as first letter or with $b$ as first letter, this implies that $u$ is a prefix of $f(a)$ and of $f(b)$.

Let $\bw$ be any LSP word beginning with $aab$ and containing the word $baaab$ (for instance the word $L_a(\bF)$ with $\bF$ the Fibonacci word).
Let $\ell$ be the greatest integer such that $u^\ell$ is a prefix of $f(a)$.
Observe $\ell \geq 1$. 
The word $\alpha u^{\ell+1}$ is a factor of $f(aa)$ and
the word $\beta u^{\ell+1}$ is a factor of $f(ba)$. 
So $u^{\ell+1}$ is a left special factor of $f(\bw)$ and so it is one of its prefixes.
By the definition of $\ell$, this implies that $f(a) = (u_1u_2)^\ell u_1$ with $u = u_1u_2$ and $u_2$ a prefix of $f(a)$.
Now using the fact that $u^{\ell+1}$ is a prefix of $f(aa)$, 
we deduce that $\alpha u^{\ell+2}$ and $\beta u^{\ell +2}$ 
are factors of $f(baaa)$ and so of $f(\bw)$:
$u^{\ell+2}$ is a left special factor  of $f(\bw)$.
Thus $f(\bw)$ begins with $u^\ell u_1u_2u_1u_2$ and with 
$f(aa) = u^\ell u_1 u_1 u_2 u^{\ell-1}u_1$
as $\ell \geq 1$.
Hence $u_1u_2 = u_2 u_1$.
By \cite[Prop. 1.3.2]{Lothaire1983book}, $\sqrt{u_1} = \sqrt{u_2}$ which implies that $\sqrt{f(a)} = \sqrt{u}$.

Exchanging the roles of $a$ and $b$, we deduce similarly that $\sqrt{f(b)} = \sqrt{u}$.
Hence  $\sqrt{f(b)} = \sqrt{f(a)}$: a contradiction with our hypotheses.
\end{proof}

From Step~2, we know that $ f = g \circ r$ for some renaming morphism $r$ and some morphism $g$ such that $\last(g(\alpha)) = \alpha$ for all letters $\alpha \in r(A)$.
Observe that $f$ preserves the LSP property for words over $A$ if and only if $g$ preserves the LSP property for words over $r(A)$.
Moreover the composition of any element of ${\cal U}$ with $r$ is, when defined, an element of ${\cal U}$. 
Thus if Proposition~\ref{P:CN_preserving2} holds for $g$, it also holds for $f$.
Hence replacing $f$ by $g$, \textit{from now on, 
we assume  that $\last(f(\alpha)) = \alpha$ for all letters $\alpha$}.

\medskip

\noindent
\textbf{Step 3}. If the images of two distinct letters begin with the same letter,
then all images of letters begin with the same letter.

\begin{proof}
Assume $k = \#A \geq 3$ and that for two distinct letters $\alpha$, 
$\beta$, $\first(f(\alpha)) = \first(f(\beta))$.
Recall that we assumed after Step~2 that for all $\delta$ in $A$, $\last(f(\delta)) = \delta$.

Now consider the morphism $g = \lambda_{a_1\ldots a_k}$ ($g(a_i) = a_1\ldots a_i$ for all $i = 1, \ldots, k$; $a_1$, \ldots, $a_k$ are pairwise distinct letters).
This morphism is bLSP.
As said just after Theorem~\ref{th_carac}  the $k$-bonacci word $\bB_{a_1, a_2, a_3, \ldots, a_k}$ contains no fragility. 
Hence
$g(\bB_{a_1, a_2, a_3, \ldots, a_k})$ is LSP.
As $f$ preserves the LSP property, $f(g(\bB_{a_1, a_2, a_3, \ldots, a_k}))$ is LSP: 
it begins with $f(a_1)$ and contains the word $f(a_1a_2a_3)$ ($\bB_{a_1, a_2, a_3, \ldots, a_k}$ begins with $a_1a_2a_1a_3$ and $g(a_3) = a_1a_2a_3$).
When $a_2 = \alpha$ and $a_3 = \beta$, denoting by $a$ the first letter of $f(\alpha)$, we see that $a_1 a$ and $\alpha a$ are factors of $f(g(\bw))$ showing that $a$ is a left special factor of $f(g(\bw))$, and so a prefix of $f(a_1)$. Consequently, $a$ is the first letter of $f(a_1)$ whatever $a_1 \in A \setminus\{\alpha, \beta\}$.
\end{proof}

\noindent
\textbf{Step 4}. If a letter is left special in $f(A)$,
then all images of letters begin with the same letter.

\begin{proof}
Here we assume that for a letter $\alpha$, there exist distinct letters $\beta$ and $\gamma$ such that
$\beta\alpha$ and $\gamma\alpha$ are both factors of images of letters. 
For any LSP word $\bw$ containing two letters whose images contains $\beta\alpha$ and $\gamma \alpha$, 
these words are factors of $f(\bw)$. As $f$ preserves the LSP property, $f(\bw)$ is LSP and so must begin with the letter $\alpha$.
As for any letter $x$, one can choose $\bw$ in such a way that it begins with $x$, 
Thus all images of letters must begin with the same letter (the letter $\alpha$).
\end{proof}

\noindent
\textbf{Step 5}. If there exists a letter $a$ such that a letter occurs twice in $f(a)$, 
then all images of letters begin with the same letter.

\begin{proof}
Let $\alpha$ be a letter occurring twice in $f(a)$.
Choose $\alpha$ such that its first occurrence is the leftmost as possible.
Assume $\alpha$ is not the first letter of $f(a)$.
If $\alpha$ is
left special in $f(a)$, then the claim is clear by Step~4.
So we can assume that all occurrences of $\alpha$ are preceded by the same letter $\beta$.
Hence we should have chosen $\beta$ instead of $\alpha$: a contradiction.
So $\alpha$ is the first letter of $f(a)$.
Let $b$ be a letter in $A$ different from $a$.
By Step~3 we only have to consider the case where $f(b)$ begins with a letter $\gamma$ different from $\alpha$.

There exists an infinite LSP word $\bw$ beginning with $a$ and containing the word $aba$ (for instance the Fibonacci word).
In $f(aba)$, only the first letter $\alpha$ of $f(\bw)$ can be left special.
But there also exist infinite LSP words beginning with $b$ and containing the word $aba$ (exchange the roles of $a$ and $b$ in the Fibonacci word).
As $f(b)$ does not begin with $\alpha$, this letter $\alpha$ cannot be left special in $f(aba)$.

Let $u$ be the word such that $|u|_\alpha = 0$ and $\alpha u \alpha$ is a prefix of $f(a)$.
As no letter is left special in the factor $f(ab)\alpha u \alpha$ of $f(aba)$,
each occurrence of $\alpha$ in $f(ab)\alpha u \alpha$ is preceded by $\alpha u$:
 $\alpha u$ is a period of $f(ab)$.
As $\alpha u$ is a prefix of $f(a)$, 
$f(ab)$ is a power of $\alpha u$. 
For some words $x$ and $y$ and some integers $k$ and $\ell$ with $k \geq 1$, $\ell \geq 0$,
we have $u = x \gamma y$, $f(a) = (\alpha x \gamma y)^k \alpha x$, 
$f(b) = \gamma y(\alpha x \gamma y)^\ell$.

Observe that $f(ba) = (\gamma y \alpha x)^{k+\ell+1}$.
Consider an LSP word $\bw$ beginning with $babbabab$ (for instance the word obtained from the Fibonacci word exchanging the roles of $a$ and $b$).
The word $f(\bw)$ contains $(f(ba))^2 = (\gamma y \alpha x)^{2(k+\ell+1)}$
and so by Lemma~\ref{lemma_x3},
$(\gamma y \alpha x)^{2k+2\ell+1}$ is a prefix of $f(\bw)$:
in particular as $k\geq 1$, $(\gamma y \alpha x)^{\ell+1+k+\ell}\gamma y\alpha$ is a prefix of $f(\bw)$.
But $f(\bw)$ begins with $f(babb)$ that begins with $(\gamma y \alpha x)^{\ell+1+k+\ell}\gamma y\gamma$. 
This contradicts the fact that $\alpha \neq \gamma$.
\end{proof}

\noindent
\textbf{Step 6}. If a letter occurs at least twice in $f(A)$, 
then all images of letters begin with the same letter.

\begin{proof}
Assume $\alpha$ occurs twice in $f(A)$.
After Step~5, we can assume that $\alpha$ occurs in $f(a)$ and $f(b)$ for two distinct letters $a$ and $b$.
For any letter $c$ (possibly $a$ or $b$),
there exists an LSP word $\bw$ beginning with $c$ and containing the word $aba$.
Thus if $\alpha$ is left special in $f(aba)$,
$\alpha$ is the first letter of $f(c)$ for all letters $c$.

Hence assume that $\alpha$ is not left special in $f(aba)$ nor in $f(bab)$.
Thus all occurrences of $\alpha$ in $f(aba)$ and in $f(bab)$ are preceded by the same letter $\beta$.
Possibly replacing iteratively $\alpha$ by $\beta$,
we see that we can choose $\alpha$ as the first letter of $f(a)$ or of $f(b)$.
Assume without loss of generality that it is the first letter of $f(a)$.
By Step~3, we can assume that $\alpha$ is not the first letter of $f(b)$.
As $\alpha$ is not left special in $f(ba)$ and as $f(b)$ ends with $b$ (so $b\alpha$ is a factor of $f(ba)$),
$b$ has two occurrences in $f(b)$.
By Step~5 all images of letters begin with the same letter.
\end{proof}

\noindent
\textbf{Step 7}. End of the proof of Proposition~\ref{P:CN_preserving2}.

The proof acts by induction on $||f|| = \sum_{a \in A} |f(a)|$.
When $||f|| = \#A$, by Step~2, $f$ is a renaming morphism: it belongs to $\cal U$.
Assume $||f|| > \#A$ and that $f$ does not belong to $\cal U$.
This means that a letter occurs at least twice in $f(A)$.
By Step~6, all images of letters begin with the same letter.
Let $a$ denote this letter.
By Step~2, for all distinct letters $x$ and $y$ in $A$, 
$\last(f(x)) \neq \last(f(y)$. 

Observe that as $f$ preserves the LSP property for infinite words, at most one letter can be left special in $f(A)$: 
this letter must be the letter $a$. 
Thus there exist an R-bLSP morphism $g_1$ and a morphism $h$,
such that $f = g_1 \circ h$ and $g_1$ is defined over ${\rm alph}(h(A))$.
Assume by contradiction that $h$ does not preserve the LSP property for infinite words. 
This means there exists an infinite LSP word $\bw$ such that $h(\bw)$ is not LSP.
By Lemma~\ref{lemma1}, $f(\bw)= g(h(\bw))$ is neither LSP: a contradiction.
Hence $h$ preserves the LSP property for infinite words.
By the definition of R-bLSP morphisms, $||g_1|| \geq \#A+1$ and so $||h|| < ||f||$.
By induction, $h \in \SRbLSP^*{\cal U}$: this also holds for $f$
 and Proposition~\ref{P:CN_preserving2} holds by induction.
\end{proof}

\subsection{Characterization of endomorphisms preserving LSP infinite words}

\begin{theorem}
\label{th2}
The set of nonerasing endomorphisms over an alphabet $A$ that preserve the LSP property for infinite words over $A$ is ${\cal L}^* Perm \cup {\cal P}_{\rm LSP}$
\end{theorem}

\begin{proof}
Permutation morphisms (elements of $Perm$) and the elements of ${\cal P}_{\rm LSP}$ preserve the LSP property for infinite words.
By Lemma~\ref{L:La1}, the elements of ${\cal L}$ also preserve this property.
Hence any element of ${\cal L}^* Perm \cup {\cal P}_{\rm LSP}$ preserves  the LSP property for infinite words.

Conversely let $f$ be a nonerasing endomorphism over $A^*$ that preserves  the LSP property for infinite words.
Assume $f \not\in {\cal P}_{\rm LSP}$.
By Lemma~\ref{lemma_x4} and the definition of ${\cal P}_{\rm LSP}$, 
for all letters $a$ and $b$, the words $f(a)$ and $f(b)$ do not have the same primitive root. 
By Proposition~\ref{P:CN_preserving2}, there exist an integer $n \geq 0$, $n$ R-bLSP morphisms $g_1$, $g_2$, \ldots, $g_n$
and an element $\nu$ of ${\cal U}$ such that $f = g_1 \circ \cdots \circ g_n \circ \nu$.

Let $B = {\rm alph}(\nu(A))$ and define the alphabets $(A_i)_{i = n, n-1, \ldots 1}$ by $A_n = {\rm alph}(g_n(B))$ and 
for $i$ such that $n-1 \geq i \geq 1$, $A_i = {\rm alph}(g_i(A_{i+1}))$. 
By Property~\ref{prop:bLSP properties}(\ref{prop:last}), 
for any $i$ in $\{n-1, \ldots, 1\}$, $A_{i+1} \subseteq A_i$
and $B \subseteq A_n$. Note that $A_1 = A$.
Hence $B \subseteq A$.
As $\nu$ belongs to ${\cal U}$, $\#B \geq \#A$. Hence $A = B$ and $\nu$ is a permutation morphism.
It follows that $f$ preserves the LSP property for infinite words if and only if $g_1 \circ \cdots \circ g_n$ preserves it.
Thus from now on we assume that $\nu$ is the identity and $f = g_1 \circ \cdots \circ g_n$.
From what precedes, we have $A_i = A$ for all $i$ in $\{n, \ldots, 1\}$. 
Hence all morphisms $g_i$ are endomorphisms. 
By definitions, the only R-bLSP morphisms that are endomorphisms are LSP morphisms.

To end we prove by induction that each endomorphism $g_i$ belongs to ${\cal L}$.
By Lemma~\ref{lemma1}, $g_n$ must preserve the LSP property for infinite words (otherwise morphisms $g_i \circ \cdots \circ g_n$ ($1\leq i \leq n$) and so $f$ would not preserve the LSP property for infinite words). 
By Lemma~\ref{L:La1}, it must belong to ${\cal L}$.
Assume we have already proved that $g_{m+1}, \ldots, g_n \in {\cal L}$ for some integer $m$, $1 \leq m < n$.
By Lemma~\ref{L:La1bis} if $g_m \not\in {\cal L}$, then $g_m$ is $(a, b, c)$-breaking for some pairwise distinct letters $a$, $b$ and $c$.
Let $\bw$ be an $(a, b, c)$-fragile LSP word. 
By Lemma~\ref{L:La2}, $g_{m+1}\circ \cdots \circ g_n( \bw)$ is $(a, b, c)$-fragile and so 
$g_{m}\circ \cdots \circ g_n( \bw)$ is not LSP.
By Lemma~\ref{lemma1}, we deduce that $g_1 \circ \cdots \circ g_m \circ \cdots g_n( \bw) = f(\bw)$ is not LSP: 
a contradiction with the fact that $f$ preserves the LSP property for infinite words. 
Then for all pairwise distinct letters $a$, $b$ and $c$, $g_m$ is not $(a, b, c)$-breaking: by Lemma~\ref{L:La1} $g_m$ belongs to ${\cal L}$. 
Hence by induction $g_i \in {\cal L}$ for all $i$, $1 \leq i \leq n$.
\end{proof}

\subsection{\label{sec:a_particular_word}A particular word}

To end the proof of Theorem~\ref{T:no simpler result}, 
we need an LSP word defined on three letters
that cannot be decomposed on two words.
Let $\xi_a$ (resp. $\xi_b$ , $\xi_c$) be the fixed point of 
$f_a = \lambda_{abc} \circ \lambda_{bca} \circ \lambda_{cab}
= [ababcababca,$ $ababcababcaab,$ $ababc]$
(resp. $f_b =  \lambda_{bca} \circ \lambda_{cab} \circ \lambda_{abc}$, 
$f_c =  \lambda_{cab} \circ \lambda_{abc} \circ \lambda_{bca}$).
Considering the permutation $\pi = [b, c, a]$, 
observe that $\pi \circ \lambda_{abc} = \lambda_{bca} \circ \pi$,
$\pi \circ \lambda_{bca} = \lambda_{cab} \circ \pi$,
and $\pi \circ \lambda_{cab} = \lambda_{abc} \circ \pi$.
Hence $\xi_a = \pi(\xi_b)$ and $\xi_b = \pi(\xi_c)$ and $\xi_c = \pi(\xi_a)$.

From Theorem~\ref{th_carac},  
$\xi_a$, $\xi_b$ and $\xi_c$ are LSP words as there exist infinite paths labeled  by 
$(\lambda_{abc} \lambda_{bca}\lambda_{cab})^\omega$, 
$(\lambda_{bca} \lambda_{cab} \lambda_{abc})^\omega$
and
$(\lambda_{cab} \lambda_{abc} \lambda_{bca})^\omega$
in the automaton ${\cal A}_{\rm bLSP}$.

\begin{lemma}
\label{L:xi_a}
There does not exist two words $u$ and $v$ such that $\xi_a$ can be decomposed on $\{u, v\}$, that is, for any choice of the words $u$ and $v$, $\xi \not\in \{u, v\}^\omega$.
\end{lemma}

\begin{proof}
First observe that $\xi_a$ is not periodic (it has infinitely many left special factors). 
Then if $X$ is a set of words such that $\xi_a \in X^\omega$, the cardinality of $X$ is at least two.
Assume by contradiction that there exist two words $u$ and $v$ such that $\xi_a \in u\{u, v\}^\omega$.
Choose $u$ and $v$ such that $|uv|$ is minimal.
Observe also that $u$ begins with the letter $a$ as it is a prefix of $\xi_a$. It is important to remark that for any word $w$ beginning by $a$, there exists a word $w'$ such that $w = \lambda_{abc}(w')$.
We distinguish three cases depending on the first letter of $v$.

\medskip

\textit{Case $v$ begins with the letter $a$}.
Then $u = \lambda_{abc}(u')$ and $v = \lambda_{abc(v')}$ and
$\xi_b \in \{u', v'\}^\omega$. As $abc$ occurs in $u$ or $v$, $|u'v'| \leq |uv|-2$.
As $\xi_a = \pi(\xi_b)$, we get $\xi_a \in \{\pi(u'), \pi(v')\}^\omega$ with $|\pi(u')\pi(v')| < |uv|$.
This contradicts the choice of $u$ and $v$.

\medskip

\textit{Case $v$ begins with the letter $b$}.
Let $I$ be the set of integers such that
$uv^iu$ occurs in $\xi_a$: 
$\xi_a \in \{uv^i \mid i \in I \}^\omega$.
Note that $\#I \geq 2$ since $\xi_a$ is not periodic.
By the definition of $\lambda_{abc}$, 
since $uv$ occurs in $\xi_a = \lambda_{abc}(\xi_b)$ and since $v$ begins with $b$,
$u$ ends with the letter $a$. This is also the case for $v$ if $I \neq \{0, 1\}$,
that is, if $vv$ occurs in $\xi_a$.

Assume that $vv$ occurs in $\xi_a$ (there exists $j \in I$ with $j \geq 2$). 
There exist words $u_1$, $v_1$ such that
$u = \lambda_{abc}(u_1)a$ and $av = \lambda_{abc}(v_1)a$.
As $\xi_a = \lambda_{abc}(\xi_b)$, $\xi_b \in \{u_1 v_1^i a \mid i \in I\}^\omega$.
As $aba$ and $abca$ must occur at least once each in $uv$, $|u_1v_1| \leq |uv|-4$.
Both words $u_1$ and $v_1$ are not empty, otherwise $\xi_a \in \{ \pi(u_1), b\}^\omega$ or $\xi_a \in \{ \pi(v_1), b\}^\omega$: a contradiction with the choice of $u$ and $v$.

Remember that $\xi_b = \lambda_{bca}(\xi_c)$.
As $v_1v_1a$ occurs in $\xi_b$ ($j \geq 2$ belongs to $I$), the word $v_1$ ends with $bc$.
and consequently the word $v_1$ begins with $a$ or $b$.

Assume that $v_1$ begins with $a$. 
Then $\xi_b \in \{ u_1 a (v_1'a)^i \mid i \in I\}^\omega$ 
where $v_1'$ is the word such that $v_1 = av_1'$.
Hence $\xi_a = \pi(\xi_b) \in \{ \pi(u_1 a) \pi(v_1'a)^i \mid i \in I\}^\omega$.
As $|\pi(u_1a)\pi(v_1'a)|=|u_1 a v_1| < |uv|$, this contradicts the choice of $u$ and $v$.

Thus $v_1$ begins with $b$. We still have $\xi_b = \lambda_{bca}(\xi_c) \in \{u_1 v_1^i a \mid i \in I\}^\omega$.
Recall that $v_1$ ends with $bc$.
There exists a word $v_2$ such that $v_1 = \lambda_{bca}(v_2)bc$.
If $v_2 = \varepsilon$, as $(bc)^3$ is not a factor of $\xi_b$, 
$I \subseteq \{0, 1, 2\}$.
Also $0 \in I$ implies that $u_1$ ends with $bc$ and
so $(bc)^3 = bc v_1^2$ is a factor of $\xi_b$.
Thus $I = \{1, 2\}$ and $\xi_b \in \{u_1bca, u_1bcbca\}^\omega$.
As $u_1$ cannot ends with $bc$, it ends with $b$ or $bca$.
Then $bbca$ or $bcabca$ is a factor of $\xi_b$ which implies that $ba$ or $aa$ is a factor of $\xi_c$. 
Both are impossible.
Hence $v_2 \neq \varepsilon$.
As $v_1^2a = (\lambda_{bca}(v_2)bc)^2a$ is a factor of $\xi_b$, 
the word $v_2cv_2a$ is a factor of $\xi_c = \lambda_{cab}(\xi_a)$.
So $v_2$ must end with the letter $c$.
As $ccc$ is not a factor of $\xi_c$, $v_2$ does not begin with $c$: its first letter is $a$.
This implies that each occurrence of $v_2$ in $\xi_c = \lambda_{cab}(\xi_a)$
is preceded by the letter $c$.
Hence each occurrence of $\lambda_{bca}(v_2)$ in $\xi_b$ is preceded by $bc$.
It follows that the word $u_1$ must end with $bc$.
This word $u_1$ also begins with $b$ as it is a prefix of $\xi_b$.
Then
there exists a word $u_2$ such that $u_1 = \lambda_{bca}(u_2)bc$.
We get $\xi_c \in \{u_2(cv_2)^ia \mid i \in I\}^\omega$.
Recall that $v_2$ begins with $a$ and ends with $c$.
There exists a word $v_3$ such that $cv_2 = \lambda_{cab}(v_3)c$.
The word $u_2$ is not empty. Indeed otherwise $\xi_a = \pi^2(\xi_c) \in \{\pi^2(cv_2), \pi^2(a)\}^\omega$ which contradicts the choice of $u$ and $v$ since $|cv_2|+|a| < |uv|$.
\begin{itemize}
\item If $0 \in I$, the factor $u_2a$ is a factor of $\xi_c$ and so the word $u_2$ ends with the letter $c$.
As $u_2$ also begins with $c$ as it is a prefix of $\xi_c$, 
there exists a word $u_3$ such that $u_2 = \lambda_{cab}(u_3)c$.
We get $\xi_a \in \{u_3(cv_3)^ia \mid i \in I\}^\omega$.
\item If $0 \not\in I$, 
for a word $u_3$, $u_2 = \lambda_{cab}(u_3)$ and
$\xi_a \in \{u_3(v_3c)^{i-1}v_3a \mid i \in I\}^\omega$.
\end{itemize}
Assume $v_3 = \varepsilon$. 
Observe that $0 \not\in I$ since $I$ contains an integer greater than or equals to $2$ and $cc$ is not a factor of $\xi_a$
Similarly $I = \{1, 2\}$.
In this case, $\xi_a \in \{u_3a, u_3ca\}^\omega$. 
As $u_3c$ is a factor of $\xi_a = \lambda_{abc}(\xi_b)$,
the word $u_3$ must end with $ab$
which implies that $u_3a$ ends with $aba$.
As $a$ is the first letter of $u_3$, $abaa$ is a factor of $\xi_a$: a contradiction.
Thus $v_3 \neq \varepsilon$. In both cases $0 \in I$ and $0 \not\in I$,
the factor $u_3v_3c$ of $\xi_a$ ends with $abc$.
As $cv_3a$ is a factor of $\xi_a$, $v_3 \neq b$.
Then $v_3$ ends with $ab$.
Once again as $u_3$ begins with $a$, we get the factor $abaa$ of $\xi_a$: a contradiction that ends the study of the case ``$vv$ occurs in $\xi_a$".

To continue the study of the case ``$v$ begins with $b$", we have to study case $I = \{0, 1\}$.
The word $u$ begins with $a$ and so there exist words $u_1$ and $v_1$ such that
$u = \lambda_{abc}(u_1)a$, and either $v = b \lambda_{abc}(v_1)$ or $v = bc \lambda_{abc}(v_1)$.
It follows that $\xi_b \in \{u_1a, u_1bv_1\}^\omega$ or $\xi_b \in \{u_1a, u_1cv_1\}^\omega$.
As $\xi_b = \lambda_{bca}(\xi_c)$, the word $u_1$ ends with $bc$. 
As $cc$ is not a factor of $\xi_b$,
necessarily $\xi_b \in \{u_1a, u_1bv_1\}^\omega$ holds. Consequently there exist words $u_2$ and $v_2$ such that
$u_1 = \lambda_{bca}(u_2)bc$ and $bv_1 = \lambda_{bca}(v_2)$: $\xi_c \in \{u_2a, u_2cv_2\}^\omega$.
Necessarily $u_2$ begins with $c$, the first letter of $\xi_c = \lambda_{cab}(\xi_a)$.
From the factor $u_2a$ we deduce that $u_2$ ends also with $c$: $u_2 = \lambda_{cab}(u_3)c$ and $cv_2 \in \lambda_{cab}(v_3)$.
It follows that $\xi_a \in \{u_3a, u_3cv_3\}^\omega$.
From the factor $u_3c$ we deduce that $u_3$ ends with $ab$. 
Also $u_3$ must begin with $a$, the first letter of $\xi_a$.
Once again we find a factor $abaa$ in $\xi_a$: a contradiction. 
This ends the study of the case ``$v$ begins with $b$".

\medskip

\textit{Case $v$ begins with the letter $c$}.
Let $I$ be defined as in the previous case: $\xi_a \in \{uv^i \mid i \in I \}^\omega$.
Assume first that $vv$ occurs in $\xi_a = \lambda_{abc}(\xi_b)$. 
Then both words $u$ and $v$ ends with $ab$.
There exist words $u_1$ and $v_1$ such that 
$u = \lambda_{abc}(u_1)ab$ and $v = c\lambda_{abc}(v_1)ab$.
Then $\xi_b \in \{u_1(cv_1)^ib \mid i \in I\}^\omega$.
As $cv_1c$ is a factor of $\xi_b = \lambda_{bca}(\xi_c)$, 
the word $v_1$ is not empty and ends with $b$.
The word $u_1c$ is a prefix of $\xi_b$ and so $u_1$ begins with $b$.
As $v_1bu_1$ is a factor of $\xi_b$, $bbb$ also occurs in $\xi_b$: a contradiction.
Hence $vv$ does not occur in $\xi_a = \lambda_{abc}$.
There exist words $u_1$ and $v_1$ such that 
$u = \lambda_{abc}(u_1)ab$ and $v = c\lambda_{abc}(v_1)$.
It follows that $\xi_b \in \{u_1b, u_1cv_1\}$.
If $u$ is empty, we find a contradiction with the choice of $u$ and $v$ 
as $|bcv_1| < |uv|$. If $u$ is not empty, it begins with $b$.
The existence of the factor $u_1c$ shows that $u_1$ also ends with $b$.
As $u_1bu_1$ is a factor of $\xi_b$, we deduce that $bbb$ is a factor of $\xi_b$: a final contradiction.
\end{proof}

\subsection{Proof of Theorem~\ref{T:no simpler result}}

Assume by contradiction that there exists a set $S$ of morphisms such 
that the set of LSP infinite words is the set of $S$-adic words.
Then all morphisms in $S$ preserve the LSP property for infinite words over $A$.

Let $\xi_a$ be the word studied in Section~\ref{sec:a_particular_word}. 
Recall that $\xi_a$ is LSP. So it is $S$-adic.
By Lemma~\ref{L:xi_a} $\xi_a$ cannot be decomposed on two words.
It follows that any element $f_n$ occurring in a directive word 
$(f_n)_{n \geq 1} \in S^\omega$ must be a morphism defined on at least three letters and cannot belong to ${\cal P}_{\rm LSP}$.
By Proposition~\ref{P:CN_preserving}, $f_n \in \SRbLSP^* {\cal U}$. If $f_n$ is defined from the alphabet $A_n$ 
to the alphabet $B_n$, then $\#A_n \leq \#B_n$. 
As $\#B_1 = 3$, by induction one can see that $\#A_n =\#B_n = 3$ for all $n \geq 1$.
This implies that $f_n = g_n \circ \pi_n$ with $g_n \in \SRbLSP^*$ 
and $\pi_n$ a renaming morphism.
As $g_n$ is an endomorphism over $B_n$, $g_n \in \SbLSP^*$.

The fact that $f_n$ preserves the LSP property for infinite words (as all elements of $S$) implies that $g_n$ also preserves the LSP property. 
Hence $g_n \in {\cal L}^*Perm$ by Theorem~\ref{th2} ($\xi_a$ is not periodic so $g_n$ cannot belong to ${\cal P}_{\rm LSP}$).

From what precedes $\xi_a = L_\alpha(\bw)$ for a letter $\alpha$ and an infinite word $\bw$.
Indeed there must exist an integer $k$ such that
$g_k \in {\cal L}^+Perm$
and $g_i \in Perm$ for all $1 \leq i \leq k-1$.
As one can see that $Perm Perm \subseteq Perm$ and
$Perm {\cal L} \subseteq {\cal L}  Perm$, relation $\xi_a = L_\alpha(\bw)$ follows.
This is a contradiction with the fact that the word $abca$ is a factor of $\xi$. 

\section{\label{sec:conclusion}Final remarks}
 
In \cite{Fici2011TCS} G.~Fici~asked for a characterization of both finite and infinite words. As explained in \cite{Richomme2017DLT}, any finite LSP word can be extended to a longer LSP word and so: A finite word is LSP if and only if it is a prefix of an infinite LSP word. And thus any characterization of infinite LSP words provides naturally a characterization of finite LSP words (adding ``is a prefix of" before the characterization of infinite LSP words). 

A natural open question comes from the content of this paper.
Does there exist an $S$-adic characterization of infinite words having at most one left special factor of each length (but not necessarily as a prefix)?
Another question comes after Lemma~\ref{L:xi_a}. The proof of this result is rather technical while the result itself seems to be extendable. 
Let $\bw$ be an infinite LSP word such that for all prefixes $p$, there exist at least three letters $a$, $b$ and $c$ such that $ap$, $bp$ and $cp$ are factors of $\bw$. Is it true that $\bw$ cannot be decomposed on two words?
 More generally let $\bw$ be an infinite word having infinitely many factors that have at least $k$ left extensions (factors $u$ such that there exist at least $k$ distinct letters $a_1$, \ldots, $a_{k}$ with $a_iu$ factor of $\bw$ for all $i$, $1 \leq i \leq k$). Is it true that $\bw$ cannot be decomposed on a set of $k-1$ or less words?

\end{document}